\documentclass[12pt,twoside]{article}

 \usepackage{float}
\usepackage{graphicx}
\usepackage{epstopdf}
\usepackage{graphicx}
\usepackage{epic}
\usepackage{multirow}
\usepackage{pst-poly}  
\usepackage{pst-plot}  
\usepackage{pst-poly}  
\usepackage{tikz}
\usepackage{xcolor}
\usetikzlibrary{arrows,shapes,chains}

\renewcommand{\paragraph}{\roman{paragraph}}
\usepackage[a4paper]{geometry}
\makeatletter
\renewcommand\title[1]{\gdef\@title{\reset@font\Large\bfseries #1}}
\renewcommand\section{\@startsection {section}{1}{\z@}%
                                   {-3.5ex \@plus -1ex \@minus -.2ex}%
                                   {2.3ex \@plus.2ex}%
                                   {\normalfont\large\bfseries}}
\renewcommand\subsection{\@startsection{subsection}{2}{\z@}%
                                     {-3ex\@plus -1ex \@minus -.2ex}%
                                     {1.5ex \@plus .2ex}%
                                     {\normalfont\normalsize\bfseries}}
\renewcommand\subsubsection{\@startsection{subsubsection}{3}{\z@}%
                                     {-2.5ex\@plus -1ex \@minus -.2ex}%
                                     {1.5ex \@plus .2ex}%
                                     {\normalfont\normalsize\bfseries}}

\def\@runningauthor{}\newcommand{\runningauthor}[1]{\def\runningauthor{#1}}
\def\@runningtitle{}\newcommand{\runningtitle}[1]{\def\runningtitle{#1}}

\renewcommand{\ps@plain}{%
\renewcommand{\@evenhead}{\footnotesize\scshape \hfill\runningauthor\hfill}
\renewcommand{\@oddhead}{\footnotesize\scshape \hfill\runningtitle\hfill}}

\newcommand{\F}{\mathbb{F}}

\newcommand {\C}{{\mathcal{C}}}

\newcommand {\ccc}{{\mathbf{c}}}
\g@addto@macro\bfseries{\boldmath}

\makeatother

\usepackage{amsthm,amsmath,amssymb}
\usepackage{cite}
\usepackage{graphicx}

\usepackage[colorlinks=true,citecolor=black,linkcolor=black,urlcolor=blue]{hyperref}

\theoremstyle{plain}
\newtheorem{theorem}{Theorem}

\newtheorem{lemma}[theorem]{Lemma}
\newtheorem{corollary}[theorem]{Corollary}
\newtheorem{proposition}[theorem]{Proposition}

\theoremstyle{definition}

\newtheorem{example}[theorem]{Example}

\theoremstyle{remark}
\newtheorem{remark}[theorem]{Remark}

\title{The $q$-ary antiprimitive BCH codes
\thanks{This research is supported by National Natural Science Foundation of China (12071001,61672036,12001175), Excellent Youth Foundation of Natural Science Foundation of Anhui Province (1808085J20), The Research Council of Norway under grant  (247742/O70) and Academic fund for outstanding talents in
universities (gxbjZD03).
}}

\runningtitle{The $q$-ary antiprimitive BCH codes}

\author{Hongwei Zhu \thanks{ School of Mathematical Sciences, Anhui University, Hefei, China. E-mail: zhwgood66@163.com}
\and Minjia Shi\thanks{School of Mathematical Sciences, Anhui University, Hefei, China. E-mail: smjwcl.good@163.com}
\and Xiaoqiang Wang\thanks{Hubei Key Laboratory of Applied Mathematics, Faculty of Mathematics and Statistics, Hubei University, Wuhan 430062, China.
email: waxiqq@163.com}
\and Tor Helleseth\thanks{Department of Informatics, University of Bergen, Bergen, Norway. E-mail: tor.helleseth@uib.no}\\
}

\runningauthor{}

\date{}

\begin{document}

\maketitle

\thispagestyle{empty}

\begin{abstract}
It is well-known that cyclic codes have efficient encoding and decoding algorithms. In recent years, antiprimitive BCH codes have attracted a lot of attention.
The objective of this paper is to study BCH codes of this type over finite fields and analyse their parameters. Some lower bounds on the minimum distance of antiprimitive BCH codes are given. The BCH codes presented in this paper have good parameters in general, containing many optimal linear codes. 
In particular, two open problems about the minimum distance of BCH codes of this type are partially solved in this paper.
\end{abstract}

\section{Introduction}
Throughout this paper, let $q$ be a power of a prime $p$ and let $\F_q$ denote a finite field with $q$ elements. A linear $[n,k,d]$ code $\C$ over $\F_q$ is a $k$-dimensional subspace of $\F_q^n$ with minimum distance $d$.
A very important class of linear codes is the class of {\it cyclic codes}, which are invariant
under cyclic shifts
\begin{equation*}
  (c_0,c_1,\ldots,c_{n-1})\longmapsto (c_{n-1},c_0,\ldots,c_{n-2})
\end{equation*}
of the coordinates.
 By identifying any vector $(c_0,c_1,\ldots,c_{n-1})\in\F_q^n$ with $c_0+c_1x+\cdots+c_{n-1}x^{n-1}\in\frac{\F_q[x]}{(x^n-1)},$ any code $\C$ of length $n$ over $\F_q$ corresponds to a subset of the principal quotient ring $\frac{\F_q[x]}{(x^n-1)}.$
 Let $\C=\langle g(x)\rangle$ be a cyclic code, where $g(x)$ is a monic polynomial of smallest degree among all the generators of $\C$. Then $g(x)$ is unique and called the {\it generator polynomial}, and $h(x)=\frac{x^n-1}{g(x)}$ is denoted the {\it parity-check polynomial} of $\C$.

Let $m={\rm ord}_n(q)$ be the order of $q$ modulo $n$, and let $\alpha$ be a generator of $\F_{q^m}^*.$ Let $\beta=\alpha^{\frac{q^m-1}{n}},$ then $\beta$ is a primitive $n$-th root of unity in $\F_{q^m}^*$. The {\it minimal polynomial} $\mathbb{M}_{\beta^s}(x)$ of $\beta^{s}$ over $\F_q$ is the monic polynomial of smallest degree over $\F_q$ with $\beta^s$ as a root.  Let $\delta$ be an integer with $2\leq\delta\leq n$ and let $b$ be an integer. A {\it BCH code} over $\F_q$ with length $n$ and designed distance $\delta$, denoted by $\C_{(q,n,\delta,b)},$ is a cyclic code with generator polynomial
\begin{equation*}\label{g}
  g_{(q,n,\delta,b)}(x)={\rm lcm}(\mathbb{M}_{\beta^{b}}(x),\mathbb{M}_{\beta^{b+1}}(x),\ldots,\mathbb{M}_{\beta^{b+\delta-2}}(x))
\end{equation*}
where the least common multiple is computed over $\F_q[x]$.
 When $b=1$, the code $\C_{(q,n,\delta,1)}$ is called a {\it narrow-sense BCH code}. If $n=q^l-1$, then $\C_{(q,n,\delta,b)}$ is referred to as a {\it primitive BCH code}.
 If $n=q^l+1$, it is called an {\it antiprimitive BCH code} by Ding in \cite{D55}.

BCH codes have been extensively studied in \cite{Aug,Augot,Ber,Bose,Charp,Charp1,DingL10,DingA,DingB,Desa,Fujiwa,Hocq,Hel,Kra,Ker,Kas,Kas1,Kas2,Kas3,Li1,LiL18,Liu1,LY1,Schoof,YueL26,YueL27,Zhu}.
Nevertheless, their parameters are known for
only a few special classes. As pointed out by Charpin in \cite{Charp} and Ding in \cite{D55},
the dimension and minimum distance of BCH codes are difficult to determine in general.
Until now, we have very limited knowledge of BCH codes, as the dimension and minimum distance of BCH codes are open in general. For a given alphabet, the dimension and minimum distance of a BCH code are known only for some special lengths and designed distances.

 In recent years, many researchers have paid considerable attention to antiprimitive BCH codes. In \cite{Li1,Liu1,YanL24}, the parameters of some classes of antiprimitive BCH codes are determined, and there are indeed many codes which have nice parameters. Linear complementary dual (LCD for short) codes are linear codes that intersect with their dual trivially.
   LCD codes can be used to protect against side-channel attacks and fault noninvasive attacks \cite{Car}.
    In particular, the antiprimitive BCH codes are LCD codes.
    Liu et al. studied the dimension of binary and nonbinary antipritimive BCH codes for some $m$ and $\delta$ in \cite{LY1,LY2}. It is worth mentioning that Ding and Tang also proved that some antiprimitive BCH codes hold several families of $t$-designs in \cite{D5,Tang1}. This has recently renewed the interest in determining the parameters of these codes.

The first objective of this paper is to give a  necessary and sufficient condition for $1\leq a\leq q^m$ being a
 coset leader modulo $q^m+1$, which is useful for us to study antiprimitive BCH codes with small dimension and to analyse their parameters.
 The second objective of this paper is to consider the minimum distance of antiprimitive BCH codes. To investigate the optimality
 of the codes studied in this paper, we compare them with the tables of the best known linear codes maintained in \cite{Gra}, and
some of  the proposed codes are optimal. In particular, two open problems recently proposed are partially solved in this paper.

The rest of this paper is organized as follows. Section II introduces some preliminaries. Section III
gives the necessary and sufficient condition for $1\leq a\leq q^m$ being a coset leader modulo $q^m+1$.
 Section IV
gives a
 lower bound on the minimum distance of $\C_{(q,q+1,\delta,1)}$, a sufficient condition for antiprimitive BCH codes to be MDS and some improved lower bounds on the minimum distance of the codes $\C_{(q,q^l+1,\delta,0)}$ and the codes $\C_{(q,q^l+1,\delta,b)}$, respectively. Two open problems in \cite{D5,Li1} are also partially solved in this section. Section V concludes the paper.


\section{Preliminaries}
In this section, we introduce some
basic concepts and known results on BCH codes, which will be used later in this paper.
\subsection{Notation used starting from now on}
Starting from now on, we adopt the following notation unless otherwise stated:
\begin{itemize}
\item $n=q^m+1$ and $\delta_i$ is the $i$-th largest coset leader modulo $n$.
\item ${\rm Tr}_q^{q^u}(\cdot)$ is the trace function from $\F_{q^u}$ to $\F_q$.
\item $\lfloor x \rfloor$ denotes the largest integer less than or equal to $x$.
\item $\beta$ is a primitive $n$-th root of unity in the finite field $\F_{q^{2m}}$.
\item  $m_i(x)$ denotes the minimal polynomial of $\beta^i$ over $\F_q$, where $1\leq i\leq n$.
\end{itemize}
\subsection{Several important linear codes and bounds}
For a vector $\ccc=(c_1,c_2,\ldots,c_n)$ in $\F_q^n$, we denote by
$wt(\ccc)=d(\ccc,\mathbf{0})$,
the {\it Hamming weight} of $\ccc$, where $d(\ccc,\mathbf{0})$ denotes the distance between $\ccc$ and $\mathbf{0}$.
The
{\it support} of a vector is the set of coordinates where the vector is
nonzero. We denote it as
$supp(\ccc)=\{i\in\{1,2,\ldots,n\}|c_i\neq0\}$.
In particular, $|supp(\ccc)|=wt(\ccc)$.

Let $\C$ be an $[n,k,d]$ code over $\F_q$, where $q=r^h$ for some prime power $r$ and some positive integer $h$. The {\it subfield subcode} of $\C$ over $\F_r$, denoted by $\C|_{\F_r}$, is defined by
$\C|_{\F_r}=\{\ccc\in \C: \ccc\in \F_r^n\}=\C\cap\F_r^n.$

If $\C$ is a linear code of length $n$ over $\F_q$, its dual code is defined by $\C^{\perp}=\{\mathbf{x}\in\F_q^n|(\mathbf{x},\mathbf{y})=\mathbf{x}\cdot\mathbf{y}^T=0~{\rm for~ all}~\mathbf{y}\in \C\},$ where $\mathbf{y}^T$ denotes the transposition of the vector $\mathbf{y}=(y_1,\ldots,y_n).$ A linear code $\C$ is called {\it linear complementary dual (LCD for short)} if $\C\bigcap \C^{\perp}=\{\mathbf{0}\}$. Let $f(x)=f_hx^h+f_{h-1}x^{h-1}+\cdots+f_1x+f_0$ be a polynomial over $\F_q$ with $f_h\neq0$ and $f_0\neq0$. The {\it reciprocal} $f^*(x)$ of $f(x)$ is defined by $f^*(x)=f_0^{-1}x^hf(x^{-1})$. A polynomial is {\it self-reciprocal} if it coincides with its reciprocal.
A code $\C$ is called {\it reversible} if $(c_0,c_1,\ldots,c_{n-1})\in \C$ implies that $(c_{n-1},c_{n-2},\ldots,c_0)\in\C.$
The following lemma which are known in \cite{Yang} and \cite[p.206]{MacW}
exhibits a relation between an LCD code and its generator polynomial $g(x)$.
\begin{lemma}\cite{Yang},\cite[p.206]{MacW}\label{LCD}
Let $\C$ be a cyclic code of length $n$ over $\F_q$ with generator polynomial $g(x)$. Then the following statements are equivalent.
\begin{itemize}
  \item $\C$ is an LCD code.
  \item $g(x)$ is self-reciprocal.
  \item $\beta^{-1}$ is root of $g(x)$ for every root $\beta$ of $g(x)$ over the splitting field of $g(x)$.
\end{itemize}
Furthermore, if $-1$ is a power of $q$ mod $n$, then every cyclic code over $\F_q$ of length $n$ is reversible.
\end{lemma}
It then follows from Lemma \ref{LCD} that every cyclic code of length $q^m+1$ over $\F_q$ is reversible. Then the antiprimitive BCH codes are LCD.

An $[n,k,n-k+1]$ linear code is called {\it maximum distance separable}, abbreviated MDS. An $[n,k,n-k]$ linear code is said to be {\it almost maximum distance separable} (AMDS for short). A code is said to be {\it near maximum distance sparable} (NMDS for short) if the code and its dual code are both AMDS.

The minimum distance of BCH codes has the following well-known bound.

\begin{lemma}\label{lemma2}
(The BCH bound) Let $\mathcal{C}$ be a cyclic code with generator polynomial $g(x)$ such that for some integers $b$ and $\delta\geq 2$,
$$g(\gamma^b)=g(\gamma^{b+1})=\cdots=g(\gamma^{b+\delta-2})=0,$$
i.e. the code has a string of $\delta-1$ consecutive powers of $\gamma$ as zeroes. Then the minimum distance of the code is at least $\delta$.
\end{lemma}

For some BCH codes, we have the following bound, which is much better than Lemma \ref{lemma2} when $\delta$ is getting large.

\begin{lemma}\label{lem:disance}\cite[Theorem 17]{Li1}
 Let $\C_{(q,q^m+1,\delta,b)}$ denote
the cyclic code of length $q^m+1$ with generator polynomial $g_{(q,q^m+1,\delta,b)}(x)$,
 then the code $\mathcal{C}_{(q,q^m+1,\delta,0)}$ has minimum distance $d\geq 2(\delta-1)$.
\end{lemma}

\subsection{Known results on parameters of $\C_{(q,q^m+1,\delta,1)}$ and $\C_{(q,q^m+1,\delta,0)}$ }
Some known results on the parameters of the codes $\C_{(q,q^m+1,\delta,1)}$ and $\C_{(q,q^m+1,\delta,0)}$ in \cite{Li1,Liu1} are listed by the following lemmas.
\begin{lemma}\cite{Li1}\label{Li11}
Let $q$ be a prime power and $m$, $\delta$ be positive integers. Then the code $\C_{(q,q^m+1,\delta,0)}$ has minimum distance $d\geq 2(\delta-1)$.
Moreover, for any integer $\delta$ with $3\leq \delta\leq q^{\lfloor\frac{m-1}{2}\rfloor}+3$, the reversible code $\C_{(q,q^m+1,\delta,0)}$ has parameters $$\left[q^m+1,q^m-2m\left(\delta-2-\left\lfloor\frac{\delta-2}{q}
\right\rfloor
\right),d\geq 2(\delta-1)\right].$$
\end{lemma}
\begin{lemma}\cite{Liu1}\label{Liu11}
Let $m\geq 3$ be an integer and $h=\left\lfloor\frac{m-1}{2}\right\rfloor$. There are the following results for $2\leq\delta\leq q^{h+1}$:
\begin{itemize}
  \item If $m\geq 4$ is an even integer, then $\C_{(q,q^m+1,\delta,1)}$ has parameters
      $$\left[q^m+1,q^m+1-2m
      \left(\delta-1-\left\lfloor\frac{\delta-1}{q}
      \right\rfloor\right),d\geq \delta\right]$$ and $\C_{(q,q^m+1,\delta+1,0)}$ has parameters $$\left[q^m+1,q^m-2m\left(\delta-\left\lfloor\frac{\delta}{q}\right\rfloor\right),d\geq 2\delta\right].$$
  \item If $m\geq3$ is an odd integer, for $2\leq \delta\leq q^{h+1}$, $\C_{(q,q^m+1,\delta,1)}$ has dimension
  \begin{equation*}
    k=\left\{
        \begin{array}{ll}
          q^m+1-2m\left(\delta-1-\left\lfloor\frac{\delta-1}{q}
          \right\rfloor\right), & \hbox{if $\delta\leq q^{h+1}-q$;} \\
          q^m+1-2m\left(q^{h+1}-q-\left\lfloor\frac{\delta-1}{q}
          \right\rfloor\right), & \hbox{if $q^{h+1}-q+1\leq\delta\leq q^{h+1}$.}
        \end{array}
      \right.
  \end{equation*}
\end{itemize}
\end{lemma}
\section{A necessary and sufficient condition for $0\leq a\leq q^m$ being a coset leader }\label{sec-MDSqeven}

Let $\mathbb{Z}_L$ denote the ring of integers modulo $L$.
Let $s$ be an integer with $0\leq s<L$. The {\it $q$-cyclotomic coset} of $s$ is defined by
$$C_s=\{s,sq,sq^2,\cdots,sq^{\ell_{s-1}}\}\,\, {\text\,\,mod \,\,L\subseteq \mathbb{Z}_L, }$$
where $\ell_s$ is the smallest positive integer such that $s\equiv sq^{\ell_s} \pmod L$, and is the size of the $q$-cyclotomic coset.
 The smallest integer in $C_s$ is called the {\it coset leader} of $C_s$.

It is known that coset leaders provide information on the Bose distance and dimension of BCH codes. In this section, we give a necessary and sufficient condition for $0\leq a\leq q^m$ being a coset leader. We call that $x$ is a coset leader means that $x$ is the coset leader of $C_x$ modulo $n$. Firstly, we consider the case that $q$ is an even prime power.

\begin{proposition}\label{dim:constaq=2}
Let $m\geq 2$ and $q$ be an even prime power. Let $1\leq i<m$, $l$ and $h$ be integers satisfying
$$1\leq l\leq \left\lfloor \frac{(q^i-1)q}{2(q+1)}\right\rfloor\,\,
\text{and}\,\,-\frac{l(q^{m-i}-1)}{q^i+1}< h<\frac{l(q^{m-i}+1)}{q^i-1}.$$  Then $0\leq a\leq q^m$ is a coset leader if and only if
$0\leq a\leq \left\lfloor\frac{q^{m+1}+q}{2(q+1)}\right\rfloor$ and $a\neq lq^{m-i}+h$.
\end{proposition}

\begin{proof}
It is easy to see that $a=0$ is a coset leader. In the following, we only consider the case $0<a\leq q^m$.
 By definition, $a$ is not a coset leader if and only if
there exist an integer $0\leq b< a$ and a positive integer $i$ such that
\begin{equation}\label{eq:cyc-01}
aq^i\equiv b \pmod{n}.
\end{equation}
Since $0\leq b<a\leq q^m$, we can assume that $1\leq i\leq 2m$.
Firstly, we consider the easy case $i=m$. Note that
$
aq^m\equiv -a \pmod{n},
$
then Eq.(\ref{eq:cyc-01}) can be written as
$
a+b\equiv 0 \pmod{n},
$
which is the same as
\begin{equation}\label{eq:abcq}
a+b=n.
\end{equation}
 It is easily seen that Eq.(\ref{eq:abcq}) holds if and only if $\frac{q^m}{2}+1\leq a\leq q^m$ since $ b< a$,
  then $a$ is not a coset leader if $\frac{q^m}{2}+1\leq a\leq q^m$.

Hence, in order to obtain the desired claims,
we only need to consider a necessary and sufficient condition for $a$ being not a coset leader when $0<a\leq \frac{q^m}{2}$ and $i\neq m$.
Obviously, Eq.(\ref{eq:cyc-01}) is equivalent to
$aq^i= b+nl_0,$
which is the same as
\begin{equation}\label{eq:cyc-02}
b+q^ml_0+l_0-aq^i=0,
\end{equation}
where $l_0$ is an integer. The proof will be carried out by distinguishing the following two cases.

\noindent {\bf Case 1}. $1\leq i<m$. In this case, Eq.(\ref{eq:cyc-02}) becomes
$b=(a-l_0q^{m-i})q^i-l_0.
$
Since $b\geq 0$, $a$ can be expressed as $a=l_0q^{m-i}+h_0$, where $h_0$ is a positive integer.
 Hence, $b=h_0q^i-l_0$. From $0\leq b<a\leq \frac{q^m}{2}$, we have
\begin{eqnarray*}
\left\{
\begin{array}{l}
l_0q^{m-i}+h_0\leq  \frac{q^m}{2},\\
h_0q^i-l_0\geq 0,\\
h_0q^i-l_0<l_0q^{m-i}+h_0.
\end{array}
\right.
\end{eqnarray*}
Solving these inequalities, we obtain $1\leq l_0\leq \frac{q^i}{2}-1$ and $1\leq h_0<\frac{l(q^{m-i}+1)}{q^i-1}$.
Then Eq.(\ref{eq:cyc-01}) holds if $a=l_0q^{m-i}+h$, where
$1\leq l_0\leq \frac{q^i}{2}-1$ and $1\leq h_0<\frac{l(q^{m-i}+1)}{q^i-1}$.

\noindent {\bf Case 2}. $m< i< 2m$. It is clear that
$aq^i \equiv -aq^{i-m} \pmod{n}.$
Then $aq^i$ can be expressed as $aq^i=t(q^m+1)-aq^{i-m}$, where $t$ is a positive integer. Let $l_1=t-l_0$, then Eq.(\ref{eq:cyc-02}) becomes
\begin{equation}\label{eq:cyc=1}
b=nl_1-aq^{i-m}.
\end{equation}
It is obvious that $l_1>0$ since $b>0$.
By definition and Eq.(\ref{eq:cyc=1}) we have
\begin{eqnarray*}
\left\{
\begin{array}{l}
a>nl_1-aq^{i-m}, \\
nl_1-aq^{i-m}>0.
\end{array}
\right.
\end{eqnarray*}
Solving these inequalities, the range of the value of $a$ can be given as
\begin{equation}\label{eq:ddibm-01}
\frac{l_1n}{q^{i-m}+1}<a< \frac{l_1n}{q^{i-m}}.
\end{equation}
It is easily seen that
\begin{equation}\label{eq:ddibm-02}
\frac{l_1n}{q^{i-m}+1}=\frac{l_1(q^{2m-i}(q^{i-m}+1)-q^{2m-i}+1)}{q^{i-m}+1}=l_1q^{2m-i}-\frac{l_1(q^{2m-i}-1)}{q^{i-m}+1}
\end{equation}
and
\begin{equation}\label{eq:ddibm-03}
\frac{l_1n}{q^{i-m}}=l_1q^{2m-i}+\frac{l_1}{q^{i-m}}.
\end{equation}
Then we have $\left\lfloor l_1q^{2m-i}+\frac{l_1}{q^{i-m}}\right\rfloor\leq \frac{q^m}{2}$ since $a\leq \frac{q^m}{2}$. Hence, $l_1\leq \frac{q^{i-m}}{2}$.

From the inequality in (\ref{eq:ddibm-01}), Eqs.(\ref{eq:ddibm-02}) and (\ref{eq:ddibm-03}), we know that $a$ can be written as $a=l_1q^{2m-i}-h_1$, where
 $1 \leq l_1\leq \frac{q^{i-m}}{2}$ and $0\leq h_1<\frac{l_1(q^{2m-i}-1)}{q^{i-m}+1}$.
Hence, Eq.(\ref{eq:cyc-01}) holds if $a=l_1q^{2m-i}-h_1$, where
 $1 \leq l_1\leq \frac{q^{i-m}}{2}$ and $0\leq h_1<\frac{l_1(q^{2m-i}-1)}{q^{i-m}+1}$.

Summarizing the discussions above, we know that $a$ is a coset leader  if and only if $a$ does not satisfy any of the following conditions:
\begin{itemize}
\item[(1)] $\frac{q^m}{2}+1\leq a\leq q^m$;
\item[(2)] $a=l_0q^{m-i}+h_0$, where $1\leq i<m$, $1\leq l_0\leq \frac{q^i}{2}-1$ and $1\leq h_0<\frac{l_0(q^{m-i}+1)}{q^i-1}$;
\item[(3)] $a=l_1q^{2m-i}-h_1$, where $m< i< 2m$, $1\leq l_1\leq \frac{q^{i-m}}{2}$ and $0\leq h_1<\frac{l_1(q^{2m-i}-1)}{q^{i-m}+1}$.
\end{itemize}
Obviously, $a$ does not satisfy Condition (2) and Condition (3) if and only if $a$ does not satisfy any of the following conditions:
\begin{itemize}
\item[(I)]  $a=lq^{m-i}+h$, where $1\leq i<m$, $1\leq l\leq \frac{q^i}{2}-1$ and $-\frac{l(q^{m-i}-1)}{q^i+1}< h<\frac{l(q^{m-i}+1)}{q^i-1}$;
\item[(II)] $a=q^{m}/2-h$, where $m< i< 2m$ and $0\leq h<\frac{q^m-q^{i-m}}{2(q^{i-m}+1)}$.
\end{itemize}
For any $m+2\leq i< 2m$, we can check that
$$\frac{q^m}{2}-\frac{q^m-q}{2(q+1)}<\frac{q^m}{2}- \frac{q^m-q^{i-m}}{2(q^{i-m}+1)}.$$
 Then the range of the value of $a$ is $$\frac{q^{m+1}+q}{2(q+1)}<a\leq \frac{q^m}{2},$$
if $a$ satisfies Condition (II).
Hence, from the discussions above, we know that $a$ is a coset leader if and only if $a$ does not satisfy any of the following conditions:
\begin{itemize}
\item $\frac{q^{m+1}+q}{2(q+1)}< a\leq q^m$;
\item $a=lq^{m-i}+h$, where $1\leq i<m$,
$1\leq l\leq \frac{q^i}{2}-1$ and $-\frac{l(q^{m-i}-1)}{q^i+1}< h<\frac{l(q^{m-i}+1)}{q^i-1}$,
\end{itemize}
which is the same as that $a$ is a coset leader if and only if
\begin{itemize}
\item $a\leq \left\lfloor\frac{q^{m+1}+q}{2(q+1)}\right\rfloor$;
\item $a\neq lq^{m-i}+h$, where $1\leq i<m$,
$1\leq l\leq \frac{q^i}{2}-1$ and $-\frac{l(q^{m-i}-1)}{q^i+1}< h<\frac{l(q^{m-i}+1)}{q^i-1}$.
\end{itemize}
Since $a\leq \left\lfloor\frac{q^{m+1}+q}{2(q+1)}\right\rfloor$, we have
\begin{equation}\label{eq:limit}
lq^{m-i}+\frac{l(q^{m-i}+1)}{q^i-1}\leq \left\lfloor\frac{q^{m+1}+q}{2(q+1)}\right\rfloor.
\end{equation}
Solving the inequality in (\ref{eq:limit}), we have $l\leq \left\lfloor \frac{(q^i-1)q}{2(q+1)}\right\rfloor$.
Hence, the range of the value of $l$ can be restricted to $1\leq l\leq \left\lfloor \frac{(q^i-1)q}{2(q+1)}\right\rfloor$.
The desired claims then follows.
\end{proof}

\begin{example}
Let $q=8$ and $m=2$. If $0\leq a\leq 65$, from Proposition \ref{dim:constaq=2}, we have that $a$ is a coset leader if and only if
$$a \in \{0,1,2,3,4,5,6,7,10,11,12,13,14,19,20,21,28\}.$$
 This result is verified by Magma programs.
\end{example}

We here give an easy application of Proposition \ref{dim:constaq=2}.

\begin{lemma}\label{dim:evenm=2}
Let $m=2$ and $q>2$, then $\delta_1=\frac{q^2-q}{2}$ and $\delta_i=\frac{q^2-3q+6-2i}{2}$, where $2\leq i\leq 4$.
\end{lemma}
\begin{proof}
From Proposition \ref{dim:constaq=2}, it is easily seen that $0\leq a\leq q^2$ is a coset leader if and only
if $0\leq a\leq \left\lfloor\frac{q^3+q}{2(q+1)}\right\rfloor$ and $a\neq lq+h$, where $1\leq l\leq\frac{q-2}{2}$ and $-\frac{l(q-1)}{q+1}< h< \frac{l(q+1)}{q-1}$.
Obviously, we have
$$\frac{l(q-1)}{q+1}=\frac{l(q+1)-2l}{q+1}=l-\frac{2l}{q+1}>l-1$$
and
$$\frac{l(q+1)}{q-1}=\frac{l(q-1)+2l}{q-1}=l+\frac{2l}{q-1}<l+1,$$
then $-\frac{l(q-1)}{q+1}\leq h\leq \frac{l(q+1)}{q-1}$ is the same as $1-l\leq h\leq l$ since $h$ is an integer.
Hence, we obtain $$\delta_1=\left\lfloor\frac{q^3+q}{2(q+1)}\right\rfloor=\left\lfloor\frac{q^2(q+1)-q(q+1)-2q}{2(q+1)}\right\rfloor=\frac{q^2-q}{2}$$
and
$\delta_i=\frac{q^2-3q+6-2i}{2}$, where $2\leq i\leq 4$.
\end{proof}

The following is a well-known result, which will be used to determine the dimension of BCH codes.
\begin{lemma}\label{lemma5}
Let $u,v,a$ be positive integers, then
\[\gcd(a^u+1,a^v-1)=\left\{ \begin{array}{lll}
           1, & \, \,\, \text{if $v_2(v)\leq v_2(u)$ and $a$ is even}, \\
           {2}, & \, \,\, \text{if $v_2(v)\leq v_2(u)$ is odd and $a$ is odd}, \\
           {a^{\gcd(u,v)}+1}, & \, \,\, \text{if $v_2(v)> v_2(u)$ is even}. \end{array}  \right.\]
\end{lemma}
\begin{lemma}\label{m=2coset}
Let $\delta_1$, $\delta_2$, $\delta_3$ and $\delta_4$ be given as in Lemma \ref{dim:evenm=2}. If $a\in \{\delta_1, \delta_2, \delta_3, \delta_4\}$, then $|C_a|=4$.
\end{lemma}
\begin{proof} It is known that ord$_n(q)=4$ when $m=2$, then $|C_{a}|$ is a divisor of $4$. Assume that $|C_{a}|=t$, then
$(q^2+1)\,|\,\delta_i(q^t-1)$. From Lemma \ref{lemma5}, it is easy to get the desired results.
\end{proof}

From Lemma \ref{lem:disance} and the previous conclusions, one can get the following results.
\begin{theorem}\label{Theorem-01}
Let $\delta_1$, $\delta_2$, $\delta_3$ and $\delta_4$ be given as in Lemma \ref{dim:evenm=2}.
Let $\delta$ be an integer. Then the BCH code $\mathcal{C}_{(q,q^2+1,\delta+1,0)}$ has parameters
 $[q^2+1,4i,d\geq 2\delta_i]$ if $\delta_{i+1}+1\leq \delta \leq \delta_i~(i=1,2,3)$ and $\mathcal{C}_{(q,q^2+1,\delta_4+1,0)}$
  has parameters $[q^2+1,16,d\geq 2\delta_4]$.
\end{theorem}

\begin{example}\label{example-01}
The following numerical examples were calculated by Magma programs, which coincide with Theorem \ref{Theorem-01}.
\begin{itemize}
\item Let $q=4$ and $4\leq \delta\leq 6$, then the code $\mathcal{C}_{(4,17,\delta+1,0)}$ has parameters $[17,4,12]$.
\item Let $q=4$ and $\delta=3$, then the code $\mathcal{C}_{(4,17,4,0)}$ has parameters $[17,8,6]$.
\item Let $q=8$ and $22\leq\delta\leq 28$, then the code $\mathcal{C}_{(8,65,\delta+1,0)}$ has parameters  $[65,4,56]$.
\item Let $q=8$ and $\delta=19$, then the code $\mathcal{C}_{(8,65,20,0)}$ has parameters $[65,16,38]$.
\end{itemize}
Note that all the four codes are optimal according to the tables of best known codes in \cite{Gra}.
\end{example}

The odd prime power case is different from the even prime power case, and has to be treated separately.  In the following,
 we consider the case that $q$ is an odd prime power.

\begin{proposition}\label{dim:qoddprime}
Let $q$ be an odd prime power. Let $1\leq i<m$, $l$ and $h$ be integers satisfying
$$1\leq l\leq \frac{q^i-1}{2}\,\,
\text{and}\,\,-\frac{l(q^{m-i}-1)}{q^i+1}< h<\frac{l(q^{m-i}+1)}{q^i-1}.$$  Then $0\leq a\leq q^m$ is a coset leader if and only if
$ a\leq \frac{n}{2}$ and $a\neq lq^{m-i}+h$.
\end{proposition}
\begin{proof}
It is easy to see that $a=0$ is a coset leader. In the following, we only consider the case $0<a\leq q^m$.
Similar to the discussion in Proposition \ref{dim:constaq=2}, we have that $a$ is not a coset leader if and only if
there exist an integer $0\leq b< a$ and a positive integer $i$ such that
\begin{equation}\label{eq:cyc-qodd}
aq^i\equiv b \pmod{n},
\end{equation}
where $1\leq i\leq 2m$. Clearly, Eq.(\ref{eq:cyc-qodd}) can be written as
$a+b=n$
if $i=m$. Then $a$ is not a coset leader if $\frac{n}{2}< a\leq q^m$.
For any positive integer $i$, Eq.(\ref{eq:cyc-qodd}) can be expressed as
\begin{equation}\label{eq:cyc-02odd}
b+q^ml_0+l_0-aq^i=0,
\end{equation}
where $l_0$ is an integer.
We now prove a necessary and sufficient condition for $a$ being not a coset leader when $a\leq \frac{n}{2}$ and $i\neq m$
by distinguishing the following two cases.

\noindent {\bf Case 1}. If $1\leq i<m$, then Eq.(\ref{eq:cyc-02odd}) becomes
\begin{equation*}
b=(a-l_0q^{m-i})q^i-l_0.
\end{equation*}
Since $b\geq 0$, $a$ can be expressed as $a=l_0q^{m-i}+h_0$, where $h_0$ is a positive integer. Hence, $b=h_0q^i-l_0$.
 From $0\leq b<a\leq \frac{n}{2}$, we obtain
\begin{eqnarray*}
\left\{
\begin{array}{l}
l_0q^{m-i}+h_0\leq \frac{n}{2},\\
h_0q^i-l_0\geq 0,\\
h_0q^i-l_0<l_0q^{m-i}+h_0.
\end{array}
\right.
\end{eqnarray*}
Solving these inequalities, we have $1\leq l_0\leq \frac{q^i-1}{2}$ and $1\leq h_0<\frac{l_0(q^{m-i}+1)}{q^i-1}$.
Then $a=l_0q^{m-i}+h_0$ is not a coset leader, where
$1\leq l_0\leq \frac{q^i-1}{2}$ and $1\leq h_0<\frac{l_0(q^{m-i}+1)}{q^i-1}$.

\noindent {\bf Case 2}. If $m< i< 2m$, similar to the discussion in Proposition \ref{dim:constaq=2}, we have that $a$ is not a coset leader if
\begin{equation*}
l_0q^{2m-i}-\frac{l_0(q^{2m-i}-1)}{q^{i-m}+1}<a< l_0q^{2m-i}-\frac{l_0}{q^{i-m}}.
\end{equation*}
 Since $0<a\leq \frac{n}{2}$, we have $l_0< q^{i-m}$ and
\begin{equation*}
\left\lfloor l_0q^{2m-i}-\frac{l_0}{q^{i-m}}\right\rfloor\leq \frac{n}{2}.
\end{equation*}
Hence,
\begin{equation*}
\begin{split}
l_0<\frac{q^i+q^{i-m}}{2(q^m-1)}&=\frac{(q^{i-m}-1)(q^m-1)+2q^{i-m}+2q^m-1}{2(q^m-1)}\\
&=\frac{q^{i-m}-1}{2}+\frac{2(q^{i-m}+q^m)-1}{2(q^m-1)}\\
&<\frac{q^{i-m}-1}{2}+1.
\end{split}
\end{equation*}
Let $a$ be expressed as $a=l_0q^{2m-i}-h_0$, then $a$ is not a coset leader if
$0<l_0\leq \frac{q^{i-m}-1}{2}$ and $1\leq h_0<\frac{l_0(q^{2m-i}-1)}{q^{i-m}+1}$.

Summarizing all the cases, we know that $a$ is a coset leader  if and only if $a$ does not satisfy any of the following conditions:
\begin{itemize}
\item[(1)] $\frac{n}{2}< a\leq q^m$;
\item[(2)] $a= l_0 q^{m-i}+h_0$, where $1\leq i<m$, $1\leq l_0\leq \frac{q^i-1}{2}$ and $1\leq h_0<\frac{l_0(q^{m-i}+1)}{q^i-1}$;
\item[(3)]  $a= l_0 q^{2m-i}-h_0$, where $m<i <2m$, $0<l_0\leq \frac{q^{i-m}-1}{2}$ and $1\leq h_0<\frac{l_0(q^{2m-i}-1)}{q^{i-m}+1}$.
\end{itemize}
It is clear that $a$ does not satisfy Condition (2) and Condition (3) is the same as $a\neq lq^{m-i}+h$,
where $1\leq i<m$, $1\leq l\leq \frac{q^i-1}{2}$ and $-\frac{l(q^{m-i}-1)}{q^i+1}< h<\frac{l(q^{m-i}+1)}{q^i-1}$.

Hence, $a$ is a coset leader  if and only if $a$ satisfies the following two conditions:
\begin{itemize}
\item[(1)] $ a\leq \frac{n}{2}$;
\item[(2)] $a\neq l q^{m-i}+h$, where $1\leq i< m$, $1\leq l\leq \frac{q^i-1}{2}$
and $-\frac{l_0(q^{m-i}-1)}{q^i+1}< h<\frac{l(q^{m-i}+1)}{q^i-1}$.
\end{itemize}
The desired claim then follows.
\end{proof}

\begin{example}
Let $q=7$ and $m=2$. If $0\leq a\leq 50$, from Proposition \ref{dim:constaq=2}, we have that $a$ is a coset leader if and only if
$$a \in \{0,1,2,3,4,5,6,9,10,11,12,17,18,25\}.$$
 This result is verified by Magma programs.
\end{example}

 It is easy to see that the largest value of $a= l q^{m-i}+h$ for $1\leq l\leq \frac{q^i-1}{2}$
and $-\frac{l(q^{m-i}-1)}{q^i+1}< h<\frac{l(q^{m-i}+1)}{q^i-1}$ is less than
$$\frac{q^i-1}{2}\cdot q^{m-i}+\frac{q^i-1}{2}\cdot\frac{(q^{m-i}+1)}{q^i-1}=\frac{n}{2}.$$
Hence, $\delta_1=\frac{n}{2}$ for any $m$. It is also very easy to check that $|C_{\delta_1}|=1$.

 We now give an easy application of Proposition \ref{dim:qoddprime}.

\begin{lemma}\label{dim:oddm=2}
Let $m=2$, then $\delta_2=\frac{(q-1)^2}{2}$, $\delta_3=\frac{q^2-2q-1}{2}$ and $\delta_4=\frac{(q-3)(q-1)}{2}$.
\end{lemma}
\begin{proof}
Let $1\leq l\leq\frac{q-1}{2}$, it can be verified that
\begin{equation}\label{eq:q-1q+1}
\frac{l(q-1)}{q+1}=\frac{l(q+1)-2l}{q+1}=\left\lfloor l-\frac{2l}{q+1}\right\rfloor=l-1
\end{equation}
and
\begin{equation}\label{eq:q-1q+11}
\frac{l(q+1)}{q-1}=\frac{l(q-1)+2l}{q-1}=\left\lfloor l+\frac{2l}{q-1}\right\rfloor=\left\{ \begin{array}{lll}
           l+1, & {\rm if}\, \,\, l=\frac{q-1}{2}, \vspace{3mm} \\
           l, & {\rm if}\, \,\, 1\leq l\leq\frac{q-3}{2}. \end{array}  \right.
\end{equation}

From Proposition \ref{dim:qoddprime}, it is easily seen that $0\leq a\leq q^2$ is a coset leader if and only if
$0\leq a\leq \frac{q^2+1}{2}$ and $a\neq lq+h$, where $1\leq l\leq\frac{q-1}{2}$ and $-\frac{l(q-1)}{q+1}< h< \frac{l(q+1)}{q-1}$.
Then from Eqs.(\ref{eq:q-1q+1}) and (\ref{eq:q-1q+11}), we have
\begin{equation*}
  ~\left\{
     \begin{array}{ll}
       \delta_2=\frac{(q-1)^2}{2}, \\
       \delta_3=\frac{(q-1)^2}{2}-1, \\
       \delta_4=\delta_3-(q-2).
     \end{array}
   \right.
\end{equation*}
The desired conclusion then follows.
\end{proof}

\begin{lemma}
Let $\delta_2$, $\delta_3$ and $\delta_4$ be given as in Lemma \ref{dim:oddm=2}. If $a\in \{\delta_2, \delta_3, \delta_4\}$, then $|C_a|=4$.
\end{lemma}
\begin{proof} The proof is similar to that of Lemma \ref{m=2coset}, we omit the details.
\end{proof}

From Lemma \ref{lem:disance} and the previous conclusions, one can get the following results.

\begin{theorem}\label{Theorem-04}
Let $\delta_1=\frac{n}{2}$, $\delta_2$, $\delta_3$ and  $\delta_4$ be given as in Lemma \ref{dim:oddm=2}. Let $\delta$ be an integer.
Then the BCH code $\mathcal{C}_{(q,q^2+1,\delta+1,0)}$ has parameters $[q^2+1,4i-3,d\geq 2\delta_i]$ if
$\delta_{i+1}+1\leq \delta \leq \delta_i~(i=1,2,3)$ and $\mathcal{C}_{(q,q^2+1,\delta_4+1,0)}$ has parameters $[q^2+1,13,d\geq 2\delta_4]$.
\end{theorem}

\begin{example}\label{example-03}
The following numerical examples were calculated by Magma programs, and coincide with Theorem \ref{Theorem-04}.
\begin{itemize}
\item Let $q=9$, then the code $\mathcal{C}_{(9,82,32,0)}$ has parameters $[82,9,62]$.
\item Let $q=5$, then the code $\mathcal{C}_{(5,26,8,0)}$ has parameters $[26,9,14]$.
\item Let $q=7$, then the code $\mathcal{C}_{(7,50,19,0)}$ has parameters  $[50,5,38]$.
\item Let $q=7$, then the code $\mathcal{C}_{(7,50,18,0)}$ has parameters  $[50,9,34]$.
\end{itemize}
Note that all the four codes are optimal according to the tables of best known codes in \cite{Gra}.
\end{example}

\begin{remark}
In \cite{LY2}, when $m=2$, the authors determined the largest coset leader for $q$ being an even prime power, and determined the first and the second largest coset leaders for $q$ being an odd prime power.
\end{remark}

\section{The lower bounds on the minimum distance of antiprimitive BCH codes}
In this section,  our main task is to study the minimum distance of antiprimitive BCH codes. Let us start with two classes of important polynomials.
\subsection{Two classes of important polynomials}
We define $[l]:=\{1,2,\ldots,l\}$. The {\it multivariate homogeneous polynomial} of degree $r$ in $l$ variables $x_1, x_2,\ldots, x_{l}$, written $P_{r,l}(x_1,\ldots,x_l)$, is defined recursively by
\begin{equation*}\label{homo1}
  P_{r,l}(x_1,x_2,\ldots,x_l)=x_lP_{r-1,l}
  (x_1,x_2,\ldots,x_l)+P_{r,l-1}(x_1,x_2,\ldots,x_{l-1});
\end{equation*}
\begin{equation*}\label{homo2}
  P_{r,1}(x_1)=x_1^r;~~~~P_{0,l}(x_1,x_2,\ldots,x_l)\equiv1;~~~~P_{-s,l}(x_1,x_2,\ldots,x_l)\equiv0,~s>0.
\end{equation*}
Equally efficient notation is $P_{r,l}$.

The {\it elementary symmetric polynomial} (ESP for short) of degree $r$ in $l$ variables $x_1, x_2,\ldots, x_{l}$, written $\sigma_{l,r}(x_1, x_2,\ldots, x_{l})$, is defined by
\begin{equation*}\label{element}
  \sigma_{l,r}(x_1, x_2,\ldots, x_{l})=\sum\limits_{I\subseteq [l],\#I=r}\prod\limits_{j\in I}x_j.
\end{equation*}

An $m\times m$ matrix
$$V_{i_1,i_2,\ldots,i_{m-1}}(a_1,a_2,\ldots,a_m)=\left(\begin{array}{cccc}
1&1&\cdots&1\\
a_1^{i_1}&a_2^{i_1}&\cdots&a_m^{i_1}\\
\vdots&\vdots&\vdots&\vdots\\
a_1^{^{i_{m-1}}}&a_2^{^{i_{m-1}}}&\cdots&a_m^{^{i_{m-1}}}
\end{array}\right)
$$
is called a {\it generalized Vandermonde matrix}, where the elements are all different, that is $s\neq t$ implies $a_s\neq a_t.$ Moreover, it is called a {\it Vandermonde matrix}, and written simply as $V(a_1,a_2,\ldots,a_m)$, if $i_1=i_2-1=\cdots=i_{m-1}-(m-2)=1$.

The following lemma establishes a useful connection between the Vandermonde determinant and the generalized Vandermonde determinant.
\begin{lemma}\label{gvand}\cite[Sec. 338]{Muir}\cite{HVDV}
Let $D_{k_1,\ldots,k_m}(x_1,\ldots,x_{m+1})$ denote the determinant of the matrix $$\left(\begin{array}{cccc}
P_{k_j-1,2}&P_{k_j-2,3}&\cdots&P_{k_j-m,m+1}
\end{array}\right)_{1\leq j\leq m},$$
where $(P_{k_j-1,2}~~P_{k_j-2,3}~~\cdots~~P_{k_j-m,m+1})$ denotes the $j$-th row of this matrix.
Then the determinant of $V_{k_1,\ldots,k_m}(x_1,\ldots,x_{m+1})$ is equal to
\begin{equation*}\label{didi}
  D_{k_1,\ldots,k_m}(x_1,\ldots,x_{m+1})\det(V(x_1,\ldots,x_{m+1})).
\end{equation*}
\end{lemma}

Let $M_{\eta,l}(x_1,\ldots, x_{\eta+l})$ be the matrix defined as
\begin{equation*}\label{Mk}
  M_{\eta,l}(x_1,\ldots, x_{\eta+l})=
  \left(\begin{array}{cccccc}
    x_1^{-l} & \cdots &   x_{\eta}^{-l}    & x_{\eta+1}^{-l} & \cdots & x_{\eta+l}^{-l} \\
    x_1^{-l+1} & \cdots &   x_{\eta}^{-l+1}    & x_{\eta+1}^{-l+1} & \cdots & x_{\eta+l}^{-l+1} \\
    \vdots &\vdots &\vdots &\vdots &\vdots &\vdots  \\
    x_1^{-1}     & \cdots & x_{\eta}^{-1} &x_{\eta+1}^{-1} & \cdots &  x_{\eta+l}^{-1}\\
    x_1^{1}     & \cdots & x_{\eta}^{1} &x_{\eta+1}^{1}    & \cdots &  x_{\eta+l}^{1}\\
    \vdots &\vdots &\vdots &\vdots &\vdots &\vdots  \\
    x_1^{\eta}     & \cdots & x_{\eta}^{\eta} &x_{\eta+1}^{\eta}&\cdots &  x_{\eta+l}^{\eta}
\end{array}\right)
\end{equation*}
with $1\leq l\leq \eta$ and $x_1,x_2,\ldots,x_{\eta+l}$ being pairwise distinct.

According to the above auxiliary lemma, we obtain the following result.
\begin{lemma}\label{le2}
Let $l$ be a positive integer and $1\leq l\leq \eta$. Then
\begin{equation*}\label{MM}
  \det(M_{\eta,l}(x_1,\ldots, x_{\eta+l}))
=\frac{\sigma_{\eta+l,\eta}(x_1,\ldots,x_{\eta+l})}
{\sigma_{\eta+l,\eta+l}^{l}(x_1,\ldots, x_{\eta+l})}
\det(V(x_1,\ldots,x_{\eta+l})).
\end{equation*}

\end{lemma}
\begin{proof}
Observe that
$\det(M_{\eta,l}(x_1,\ldots,x_{\eta+l}))=\frac{\det(V_{1,2,\ldots,l-1,l+1,\ldots,\eta+l}(x_1,\ldots,x_{\eta+l}))}
{\sigma_{\eta+l,\eta+l}^{l}(x_1,\ldots,x_{\eta+l})}$.  Lemma \ref{didi} implies that
\begin{eqnarray*}
  \det(V_{1,2,\ldots,l-1,l+1,\ldots,\eta+l}(x_1,\ldots,x_{\eta+l}))=&~& ~~~~~~~~~~~~~~~~~~~~~~~~~~~~~~~~~ \\
  \det\left(\begin{array}{ccccc}
    P_{1,l+1}
     & P_{0,l+2}
      &   0    & \cdots &0 \\
    P_{2,l+1}
     & P_{1,l+2}
     &P_{0,l+3}
     &\cdots&0 \\
    \vdots &\vdots &\vdots &\vdots &\vdots  \\
    P_{\eta-1,l+1}
    &P_{\eta-2,l+2}
    &P_{\eta-3,l+3}
    &\cdots &P_{0,l+\eta}
    \\
    P_{\eta,l+1}
    &P_{\eta-1,l+2}
    &P_{\eta-2,l+3}
     & \cdots &P_{1,l+\eta}
\end{array}\right) &\cdot& \det(V(x_1,\ldots,v_{\eta+l})).
\end{eqnarray*}

It is easy to verify that $P_{0,l}=1$, $P_{1,l}=\sigma_{l,1}(x_1,\ldots,x_l)$ and $P_{r,l}=\sum\limits_{i_1+\cdots+i_l=r}\prod\limits_{j=1}^{l} x_j^{i_j}$, where $i_1,\ldots,i_l$ are nonnegative integers.
The following equality
\begin{equation}\label{importan}
  \det(V_{1,2,\ldots,l-1,l+1,\ldots,\eta+l}(x_1,\ldots,x_{\eta+l}))
  =\sigma_{\eta+l,\eta}(x_1,\ldots,x_{\eta+l}),
\end{equation} is proved by induction on $\eta$. The initial case $\eta=1$ is trivial. When $\eta=2$,
\begin{footnotesize}
\begin{eqnarray*}
  \det\left(V_{1,\ldots,l-1,l+1,l+2}(x_1,\ldots,x_{l+2})\right) &=& \det\left(
                                                   \begin{array}{cc}
                                                     \sigma_{l+1,1}(x_1,\ldots,x_{l+1}) & 1 \\
                                                     P_{2,l+1} & \sigma_{l+2,1}(x_1,\ldots,x_{l+2}) \\
                                                   \end{array}
                                                 \right) \\
  ~ &=& \sigma_{l+1,1}(x_1,\ldots,x_{l+1})\sigma_{l+2,1}(x_1,\ldots,x_{l+2})
-\sum\limits_{i_1+\cdots+i_{l+1}=2}\prod\limits_{j=1}^{l+1} x_j^{i_j} \\
  ~ &=& \sigma_{l+2,2}(x_1,\ldots,x_{l+2}).
\end{eqnarray*}
\end{footnotesize}
 Assuming Eq.(\ref{importan}) holds with $\eta$ replaced by $\eta-a$, where $a\in\{1,2,\ldots,\eta-1\}.$ For the sake of convenience, let
 $\det(V_{1,2,\ldots,l-1,l+1,\ldots,\eta+l}(x_1,\ldots,x_{\eta+l})):=F(\eta)$
and let $\sigma_{k,l}$ denote $\sigma_{k,l}(x_1,\ldots,x_{k})$ in this proof.
 It then follows from $F(\eta-a)=\sigma_{\eta-a+l,\eta-a}$ that
\begin{small}
\begin{eqnarray*}
  F(\eta) &=& P_{1,l+\eta}F(\eta-1)-P_{2,l+\eta-1}F(\eta-2)+\ldots+(-1)^{\eta-2}P_{\eta-1,l+2}F(1)+(-1)^{\eta-1}P_{\eta,l+1}F(0) \\
  ~ &=& P_{1,l+\eta}\sigma_{\eta-1+l,\eta-1}-P_{2,l+\eta-1}\sigma_{\eta-2+l,\eta-2}+\ldots+(-1)^{\eta-2}P_{\eta-1,l+2}\sigma_{l+1,1}
+(-1)^{\eta-1}P_{\eta,l+1} \\
  ~ &=& \sigma_{\eta+l,\eta}.
\end{eqnarray*}
\end{small}
The desired conclusion is therefore obtained.
\end{proof}

\subsection{The lower bound on the minimum distance of $\C_{(q,q+1,\delta,1)}$}
In this subsection, we give a lower bound on the minimum distance of $\C_{(q,q+1,\delta,1)}$ under some conditions. Besides, a sufficient
condition for the antiprimitive BCH codes to be AMDS (even MDS) is given.
When $l=1$, ${\rm ord}_{q+1}(q)=2$.
Let $\eta=\delta-1$, $U_{q+1}$ be the subgroup of $\F_{q^2}^*$ of order $q+1$ and ${{U_{q+1}}\choose {i}}$ denote the set of all $i$-subsets of $U_{q+1}$.

\begin{theorem}\label{TH3}
Let $q=p^s>2\eta$ with $s$ being a positive integer. Then the narrow-sense BCH code $\C_{(q,q+1,\delta,1)}$ over $\F_q$ has parameters $[q+1,q-2\eta+1,d]$, where $d\geq \eta+l$ if $\sigma_{\eta+w,\eta}(u_1,\ldots,u_{\eta+w})\neq0$ for any $\{u_1,\ldots,u_{\eta+w}\}\in {{U_{q+1}}\choose {\eta+w}}$ with $1\leq w<l$.
\end{theorem}
\begin{proof}
For convenience, let $n=q+1$ and $\gamma=\beta^{-1}$. Then $\gamma$ is a primitive $n$-th root of unity in $\F_{q^2}^*.$
Clearly, $m_i(x)$ has only roots $\beta^{i}$ and $\beta^{qi}$ with $1\leq i\leq \eta$ as $\beta^{q^2}=\beta$. By the definition of BCH codes, the generator polynomial of $\C_{(q,q+1,\delta,1)}$ is $g(x):=\prod_{i=1}^{\eta}m_{i}(x)$,  $\deg{(g(x))}=2\eta$ and the dimension of $\C_{(q,q+1,\delta,1)}$ is $q-2\eta+1$.
It then follows from Delsarte's theorem \cite{DELSARTE} that the trace expression of $\C^{\perp}_{(q,n,\delta,1)}$ is given by
\begin{equation*}\label{Delsarte}
  \C^{\perp}_{(q,n,\delta,1)}=\{\ccc_{(a_1,\ldots,a_{\eta})}:a_1,\ldots,a_{\eta}\in\F_{q^2}\}
\end{equation*}
where $\ccc_{(a_1,\ldots,a_{\eta})}=({\rm{Tr}}_{q}^{q^2}(a_1\gamma^i+a_2\gamma^{2i}+\cdots+a_{\eta}\gamma^{\eta i}))_{i=0}^{q}$.
 Define
\begin{equation*}\label{H}
  H=
  \left(\begin{array}{ccccc}
    1 & \gamma^1 &   \gamma^2    & \cdots & \gamma^q \\
    1 & \gamma^2 &   \gamma^4    & \cdots & \gamma^{2q}\\
    \vdots&\vdots&\vdots&\vdots&\vdots\\
    1 & \gamma^{\eta}&\gamma^{2\eta}&\cdots&\gamma^{q\eta}
\end{array}\right).
\end{equation*}
It is not difficult to check that $H$ is a parity-check matrix of $\C_{(q,q+1,\delta,1)}.$
 If $\ccc=(a_1,\ldots,a_{q+1})\in \C_{(q,q+1,\delta,1)}$ and $wt(\ccc)=d$, then there exists $\{\gamma^{i_1},\ldots,\gamma^{i_d}\}\in {{U_{q+1}}\choose d}$ such that $\sum_{j=1}^d a_{s_j}\gamma^{i_j}=0$, where $\{s_1,\ldots,s_d\}={supp}(\ccc)$.  The equation $\sum_{j=1}^d a_{s_j}\gamma^{i_j}=0$ is equivalent to $\sum_{j=1}^d a_{s_j}\gamma^{-i_j}\\=0$ follows from ${(\gamma^{i_j})}^q={\gamma^{-i_j}}$.

According to the BCH bound and the Singleton bound, we have that $\eta+1\leq d\leq 2\eta+1$.
Assume that $d=\eta+w<\eta+l$. Then there are $\eta+w$ pairwise distinct elements $u_1,\ldots,u_{\eta+w}\in U_{q+1}$ such that
\begin{equation}\label{import}
  a_{s_1}\left(
    \begin{array}{c}
      u_1 \\
      u_1^2 \\
      \vdots \\
      u_1^{\eta} \\
    \end{array}
      \right)
    +a_{s_2}\left(
    \begin{array}{c}
      u_2 \\
      u_2^2 \\
      \vdots \\
      u_2^{\eta} \\
    \end{array}  \right)
    +\cdots+
    a_{s_{\eta+w}}\left(
    \begin{array}{c}
      u_{\eta+w} \\
      u_{\eta+w}^2 \\
      \vdots \\
      u_{\eta+w}^{\eta} \\
    \end{array}
  \right)=0.
\end{equation} Eq.(\ref{import}) is equivalent to
\begin{equation*}
  a_{s_1}\left(
    \begin{array}{c}
    u_{1}^{-w}\\
    u_{1}^{-w+1}\\
    \vdots\\
    u_{1}^{-1}\\
      u_1 \\
      u_1^2 \\
      \vdots \\
      u_1^{\eta} \\
    \end{array}
      \right)
    +a_{s_2}\left(
    \begin{array}{c}
    u_{2}^{-w}\\
    u_{2}^{-w+1}\\
    \vdots\\
    u_{2}^{-1}\\
      u_2 \\
      u_2^2 \\
      \vdots \\
      u_2^{\eta} \\
    \end{array}  \right)
    +\cdots+
    a_{s_{\eta+w}}\left(
    \begin{array}{c}
    u_{\eta+w}^{-w}\\
    u_{\eta+w}^{-w+1}\\
    \vdots\\
    u_{\eta+w}^{-1}\\
      u_{\eta+w} \\
      u_{\eta+w}^2 \\
      \vdots \\
      u_{\eta+w}^{\eta} \\
    \end{array}
  \right)=0.
\end{equation*}
It then follows since $a_{s_1},\ldots,a_{s_{\eta+w-1}}$ and $a_{s_{\eta+w}}$ are not all zeros (in fact, $a_{s_j}\neq0$ for all $j\in\{1,2,\ldots,\eta+w\}$) that $\det(M_{\eta,w}(u_1,\ldots,u_{\eta+w}))=0$. According to Lemma \ref{le2}, we have $\sigma_{\eta+w,\eta}(u_1,\ldots,u_{\eta+w})=0$, which is contrary to our assumption.  This completes the proof.
\end{proof}
Since $\eta+1\leq d\leq 2\eta+1$, we have the following result, which is a sufficient condition for  the
antiprimitive BCH codes of this type to be AMDS (even MDS).
\begin{corollary}\label{COR8}
If  $\prod_{i=1}^{\eta}\sigma_{\eta+i,i}(u_1,\ldots,u_{\eta+i})\neq0$ for any $\{u_1,\ldots,u_{2\eta}\}\in {U_{q+1}\choose 2\eta}$, then $\C_{(q,q+1,\delta,1)}$ is MDS.
If $\prod_{i=1}^{\eta-1}\sigma_{\eta+i,i}(u_1,\ldots,u_{\eta+i})\neq0$ for any $\{u_1,\ldots,u_{2\eta-1}\}\in {U_{q+1}\choose 2\eta-1}$, then the minimum distance of $\C_{(q,q+1,\delta,1)}$ is no less than $2\eta$.
\end{corollary}
\begin{proof}
According to the Singleton bound, we have that $d\leq 2\eta+1$.
From Theorem \ref{TH3}, we have that $d\geq 2\eta+1$ under the first condition and $d\geq 2\eta$ under the second condition. This completes the proof.
\end{proof}
In other words, $\C_{(q,q+1,\delta,1)}$ is either an MDS code or an AMDS code under the second condition.
\begin{example}
Let $\{u_1,u_2,\ldots,u_6\}\in {U_{2^m+1}\choose 6}.$ From Lemma 9 and Lemma 17 in \cite{Tang1}, we have that $\sigma_{4,1}(u_1,u_2,u_3,u_4)\neq0$ for any $m$ and $\sigma_{5,2}(u_1,u_2,u_3,u_4,u_5)\neq0$ if $m$ is odd. According to Corollary \ref{COR8}, $\C_{(2^m,2^m+1,4,1)}$ is an either MDS code or an AMDS code if $m$ is odd. In fact, $\C_{(2^m,2^m+1,4,1)}$ is NMDS when $m$ is odd.
\end{example}
For $m\geq 3$, a cyclic code $\C$ of length $p^m+1$ can be defined as follows
\begin{equation*}
  N(p^m)=\left\{(c_1,c_2,\ldots,c_{p^m+1})\in \F_p^{p^m+1}:\sum_{i=1}^{p^m+1}c_i\beta^{i-1}=0\right\},
\end{equation*}
where $\beta$ is a generator of $U_{p^m+1}.$
This code $N(p^m)$ is called a {\it $p$-ary Zetterberg code} and its dual code $N(p^m)^{\perp}$ has a very simple trace description
$$N(p^m)^{\perp}=\{({\rm Tr}_p^{p^{2m}}(a\beta^0),{\rm Tr}_p^{p^{2m}}(a\beta^1),\ldots,{\rm Tr}_p^{p^{2m}}(a\beta^{p^m})):a\in\F_{p^{2m}}\}.$$
The weight distributions of binary and ternary Zetterberg codes were considered in \cite{Geer,Schoof} and their low-weight codewords over binary fields and ternary fields were obtained in \cite[Table 6.2]{Schoof} and \cite[Lemma 6]{Xia1}, respectively.
Let $B_i^p$ denote the number of its codewords with Hamming weight $i$ over $\F_p$.
If $B_i^p=0$, then $\sigma_{i,1}(u_1,u_2,\ldots,u_{i})\neq0$ for any $\{u_1,u_2\ldots,u_{i}\}\in {U_{p^m+1}\choose i}$. In order to find infinite families of AMDS or MDS codes, we pay more attention to the case that $B_i^p=0$ with $i\geq3$.

In the binary case, we observe that $B_4^2=0$ for any $m$ and $B_3^2=0$ if $m$ is even from \cite[Table 6.2]{Schoof}. In the ternary case, we see $B_3^3=0$ from \cite[Lemma 6]{Xia1}.  Ding and Tang proved that $\C_{(2^m,2^m+1,3,1)}$, $\C_{(2^m,2^m+1,4,1)}$ and $\C_{(3^m,3^m+1,3,1)}$ are indeed AMDS in some cases (in fact, they are NMDS).


In addition to the construction of AMDS codes, the ESPs can also be applied to the construction of $t$-designs.
Define
\begin{equation}\label{B}
  \mathcal{B}_{\sigma_{\eta+l,\eta,q+1}}=
  \left\{\{u_1,\ldots,u_{\eta+l}\}\in {{U_{q+1}}\choose {\eta+l}}:\sigma_{\eta+l,\eta}(u_1,\ldots,u_{\eta+l})=0\right\}.
\end{equation}
The incidence structure $\mathbb{D}_{\sigma_{\eta+l,\eta,q+1}}=(U_{q+1},\mathcal{B}_{\sigma_{\eta+l,\eta,q+1}})$ may be a $t$-design for some $\lambda$, where $U_{q+1}$ is the point set, and the incidence relation is the set membership. In this case, we say that the ESP $\sigma_{\eta+l,\eta}$ supports a $t$-$(q+1,k,\lambda)$ design. For more information on $t$-designs from linear codes, please refer to \cite{D2}. Unfortunately, all $t$-designs held in MDS codes are complete. So we are more concerned about whether AMDS codes over finite fields hold $t$-designs.

\begin{remark}
Since $\sigma_{\eta+l,\eta}(u_1,\ldots,u_{\eta+l})=\frac{\sigma_{\eta+l,l}^q(u_1,\ldots,u_{\eta+l})}
{\sigma_{\eta+l,\eta+l}(u_1,\ldots,u_{\eta+l})}$ and $\sigma_{\eta+l,\eta+l}(u_1,\ldots,u_{\eta+l})\neq0$, $\sigma_{\eta+l,\eta}(u_1,\ldots,u_{\eta+l})=0$ is equivalent to $\sigma_{\eta+l,l}(u_1,\ldots,u_{\eta+l})=0$. Tang and Ding \cite{Tang1} proposed the idea which is to construct $t$-designs from the ESP $\sigma_{k,l}(u_1,\ldots,u_k)$ with $l\leq\frac{k}{2}$ and $1\leq k\leq q+1$. However, they did not show the relationship between the ESPs and the minimum distance of $C_{(q,q+1,\delta,1)}$ for any $\delta$. Theorem \ref{TH3} explains the reason why we pay attention to the ESPs during the coding-theoretic construction. Ding and Tang studied the cases that $\eta=2$ and $3$, and presented several infinite families of linear codes holding infinite families of $t$-designs for $t=2,3,4$ in \cite{D5,Tang1}.
\end{remark}

Tang and Ding proposed an interesting open problem about the cardinality of $\mathcal{B}_{\sigma_{{k,l},q+1}}$, and considered the cases $(k,l)\in\{(6,3),(5,2),(4,1), (4,2),(3,1)\}.$

{\bf Open Problem in \cite{Tang1}}:
Let $q=2^l$, and $k, l$ be two positive integers with $l\leq \frac{k}{2}$. Determine the cardinality of the block set $\mathcal{B}_{\sigma_{{k,l},q+1}}$ given by (\ref{B}).

In fact, the cardinalities of $\mathcal{B}_{\sigma_{{k,1},2^s+1}}$ are determined by the binary Zetterberg codes, that is, $\#\mathcal{B}_{\sigma_{{k,1},2^s+1}}=B_i^2$.
It is better to revise the above open problem to ``$q$ be a prime power''. Notice that if $p>2$, then $\#\mathcal{B}_{\sigma_{{k,1},p^s+1}}\leq B_i^p.$

For any $q$ and some $\eta$, we have the following result.
\begin{proposition}\label{thif}
If $\eta+1|q+1$, then there exists $(u_1,\ldots,u_{\eta+1})\in{{U_{q+1}}\choose {\eta+1}}$ such that
\begin{equation*}\label{sigma11}
  \sigma_{\eta+1,1}(u_1,\ldots,u_{\eta+1})=\sigma_{\eta+1,\eta}(u_1,\ldots,u_{\eta+1})=0.
\end{equation*}
\end{proposition}
\begin{proof}
Let $\gamma$ be a generator of $U_{q+1}$ and $\theta=\gamma^{\frac{q+1}{\eta+1}}$.
Note that $\theta^{\eta+1}-1=(\theta-1)(\theta^{\eta}+\cdots+\theta+1)=0$. The equality $\theta^{\eta}+\cdots+\theta+1=0$ holds because of $\theta-1\neq0$. It then follows from the above remark that $\sigma_{\eta+1,\eta}(1,\theta,\theta^2,\ldots,\theta^{\eta})=0.$
\end{proof}

\subsection{The lower bounds on the minimum distance of $\C_{(q,q^m+1,\delta,b)}$ for $m\geq1$}
By the Delsarte theorem, we have that $\C_{(q^m,q^m+1,\delta,b)}|_{\F_q}=\C_{(q,q^m+1,\delta,b)}$. This equality is useful for deriving the minimum distance of $\C_{(q,q^m+1,\delta,b)}.$
Let $H$ be a parity-check matrix of $\C_{(q^m,q^m+1,\delta,b)}$. If there exists some vector $\ccc={(c_1,\ldots,c_{q^m+1})}\in \F_q^{q^m+1}$ with Hamming weight $i$ such that $\ccc H^T=\mathbf{0}$, then the minimum distance of  $\C_{(q,q^m+1,\delta,b)}$ is no less than $i$. To some extent, it will help us to determine the minimum distance of some BCH codes with small designed distance $\delta$.

The following theorem improves the lower bound of Lemma \ref{lem:disance} when $\delta\equiv1~({\rm mod}~q)$.
\begin{theorem}\label{TH10}
If $\delta\equiv1~({\rm mod}~q)$ and $n\geq 2\delta$, then
$\C_{(q,q^m+1,\delta,0)}=\C_{(q,q^m+1,\delta+1,0)}=
\C_{(q,q^m+1,2(\delta-1),q^m-(\delta-2))}=
\C_{(q,q^m+1,2\delta,q^m-(\delta-1))}$. Moreover, the code $\C_{(q,q^m+1,\delta,0)}$ has minimum distance $d\geq 2\delta$ in this case.
\end{theorem}
\begin{proof}
The proof of Theorem 17 in \cite{Li1} shows that the generator polynomial $g_{(q,q^m+1,\delta,0)}(x)$ has the zeroes $\beta^i$ for all $i$ in the set
\begin{equation*}
  \{q^m-(\delta-2),\ldots,q^m-2,q^m-1,0,1,2,\ldots,\delta-2\}.
\end{equation*}
It then follows that $\C_{(q,q^m+1,\delta,0)}=\C_{(q,q^m+1,2(\delta-1),q^m-(\delta-2))}$ for any positive integer $\delta\geq 2$.
Next, we only need to prove that $\C_{(q^m,q^m+1,\delta,0)}|_{\F_q}
=\C_{(q^m,q^m+1,\delta+1,0)}|_{\F_q}$.
In a manner similar to the proof of Theorem \ref{TH3}, we need the parity-check matrix of $\C_{(q^l,q^l+1,\delta,0)}$ when $\delta\equiv 1~({\rm mod}~q)$. For convenience, take $\delta=q^jt+1$ for some nonnegative integers $j$ and $t$, where $1<t<q$.
Let $\gamma=\beta^{-1}$. Then $\gamma$ is a primitive $(q^m+1)$-th root of unity.
Define
\begin{equation*}
  H=\left(
      \begin{array}{ccccc}
        1 & 1               &   1                  &   \cdots   &   1                               \\
        1 & \gamma          &   \gamma^2           &   \cdots   &   \gamma^{q^m}                    \\
        1 & \gamma^2        &   \gamma^4           &   \cdots   &   \gamma^{2q^m}                   \\
   \vdots & \vdots          &   \vdots             &   \vdots   &   \vdots                          \\
        1 & \gamma^{q^jt-1} &   \gamma^{2(q^jt-1)} &   \cdots   &   \gamma^{(q^jt-1){q^m}}          \\
      \end{array}
    \right).
\end{equation*}
It is easy to check that $H$ is a parity-check matrix of $\C_{(q^m,q^m+1,q^jt+1,0)}$, and
\begin{equation*}
  \C_{(q^m,q^m+1,q^jt+1,0)}|_{\F_q}
  =\{\ccc=(c_1,\ldots,c_{q^m+1})\in\F_{q}^{q^m+1}:\ccc H^{T}=\mathbf{0}\}.
\end{equation*}
 Raising to the $q$-th power both sides of the equation $c_1+c_2\gamma^{q^{j-1}t}+c_3\gamma^{2q^{j-1}t}
 +\cdots+c_{q^m+1}\gamma^{q^{m}q^{j-1}t}=0$ yields
\begin{equation*}\label{BBB}
 c_1+c_2\gamma^{q^{j}t}+c_3\gamma^{2q^{j}t}+
 \cdots+c_{q^m+1}\gamma^{q^{m}q^{j}t}=0.
\end{equation*}
Let $\mathbf{w}=(1,\gamma^{q^{j}t}, \gamma^{2q^{j}t},\cdots,\gamma^{q^{m}q^{j}t}).$
It then follows that \begin{equation*}
  \C_{(q^m,q^m+1,q^jt+1,0)}|_{\F_q}=
  \{\ccc=(c_1,\ldots,c_{q^m+1})\in\F_{q}^{q^m+1}:\ccc H_1^{T}=\mathbf{0}\},
\end{equation*}
where $H_1:=
\left(
  \begin{array}{c}
    H \\
              \mathbf{w}\\
  \end{array}
\right)
.
$ Since $H_1$ is a parity-check matrix of $\C_{(q^m,q^m+1,q^jt+2,0)}$, the desired result is obtained.
\end{proof}
The following conclusion generalizes the result in Theorem \ref{TH10} and gives a lower bound on the minimum distance of the code $\C_{(q,q^m+1,\delta,b)}$ for any $b$.
\begin{corollary}\label{COR10}
If $\delta\equiv1-b~({\rm mod}~q)$, then
$\C_{(q,q^m+1,\delta,b)}=\C_{(q,q^m+1,\delta+1,b)}$. Moreover, the minimum distance of the $\C_{(q,q^m+1,\delta,b)}$ is no less than $\delta+1.$
\end{corollary}
\begin{proof}
 The proof is similar to that of Theorem \ref{TH10}, and
is omitted here.
\end{proof}
The following lemma is useful in the proof of the subsequent theorem.
\begin{lemma}\cite[p.247]{Betten}\label{Betten}
Let $\C$ be a narrow-sense BCH code of length $n$ with designed distance $\delta$ over $\F_q$. If $\delta$ divides $n$, then the minimum distance $d=\delta$.
\end{lemma}
Inspired by Lemma \ref{Betten}, we give the following results.
\begin{theorem}\label{TH11}
Let $Char(\F_q)=2$, $m=l_0(2t+1)$ with $l_0$ and $t$ being positive integers and $\delta=q^{l_0}+1$. Then the minimum distance of $C_{(q,q^m+1,\delta,0)}$ is $2\delta.$
\end{theorem}
\begin{proof}
According to Theorem \ref{TH10}, $d\geq 2\delta.$ Then $q^m+1$ can be written in the form $(q^{l_0}+1)((q^{l_0})^{2t}+(-1)
\cdot(q^{{l_0}})^{2t-1}+\cdots+(-1)^{2t}\cdot1)$ as $m=l_0(2t+1)$.
Then there exists $\xi\in U_{q^m+1}$, where the order of $\xi$ is $q^{l_0}+1$. Take $u_i=\xi^i$ and $u_{i+q^{l_0}+1}=\xi^i\cdot \alpha$ with $i=1,\ldots,q^{l_0}+1$ and $\alpha\in U_{p^m+1}\backslash \langle\xi\rangle.$ Notice that $\langle\xi\rangle$ denotes the subgroup generated by $\xi$. We have then
$$\left\{
  \begin{array}{ll}
    \sum_{j=1}^{2(q^{l_0}+1)}u_j^0=0,\\
    \sum_{j=1}^{2(q^{l_0}+1)}u_j= 0, \\
    \sum_{j=1}^{2(q^{l_0}+1)}u_j^2=0, \\
    \cdots \\
    \sum_{j=1}^{2(q^{l_0}+1)}u_j^{q^{l_0}}= 0.
  \end{array}
\right.$$
Therefore, there exists a vector $\ccc\in \F_q^{q^m+1}$ with $wt(\ccc)=2\delta$ such that $\ccc H^T=\mathbf{0}$, where $H$ is a parity-check matrix of $\C_{(q^m,q^m+1,\delta,0)}.$
This completes the proof.
\end{proof}

Note that if $Char(\F_q)\neq2$, Theorem \ref{TH11} is not necessarily true since $b=0$.
Compared with Lemma \ref{Betten}, the designed distance $\delta$ in Theorem \ref{TH11} divides the length $q^m+1$, but $2\delta$ does not divide $q^m+1$. When $b$ is not equal to $0$, we have the following result.
\begin{corollary}\label{COR12}
Let $m=l_0(2t+1)$ with $l_0$ and $t$ being positive integers and let $\delta=q^{l_0}+1-b$ with $b\neq0$. Then the minimum distance of $C_{(q,q^m+1,\delta,b)}$ is $\delta+1.$
\end{corollary}
\begin{proof}
The proof is similar to that of Theorem \ref{TH11}, and
is omitted here. Notice that there is no limit on $Char(\F_q)$ because $b\neq0$.
\end{proof}
Combining Theorem \ref{TH11} and Corollary \ref{COR12} with  Lemma \ref{Li11} and Lemma \ref{Liu11}, the parameters of two families of antiprimitive BCH codes are completely determined.
\begin{proposition}\label{PPP}
Let $m=l_0(2t+1)$ with $l_0$ and $t$ being positive integers and let $\delta=q^{l_0}+1\geq3$. If $Char(\F_q)=2$, then the reversible code $\C_{(q,q^m+1,\delta,0)}$
has parameters
\begin{equation*}
  [q^m+1,q^m-2m(q^{l_0}-q^{l_0-1}),2q^{l_0}+2].
\end{equation*}
If $\delta=q^{l_0}\geq2$, then
the reversible code $\C_{(q,q^m+1,\delta,1)}$
has parameters
\begin{equation*}
  [q^m+1,k,\delta+1],
\end{equation*}
where $k$ is given by Lemma \ref{Liu11}.
\end{proposition}
The second result in  Proposition \ref{PPP} is trivial as $\delta+1$ divides $q^m+1$, which is a special case of Lemma \ref{Betten}.

\begin{theorem}\label{COR11}
Let $q=2$. Then the following hold:
\begin{itemize}
  \item $\C_{(2,2^m+1,3,0)}=\C_{(2,2^m+1,4,0)}$ if $m\geq3$, and the minimum distance of $\C_{(2,2^m+1,3,0)}$ is $6$.
  \item $\C_{(2,2^m+1,5,0)}=\C_{(2,2^m+1,6,0)}$ if $m\geq4$, and the minimum distance of $\C_{(2,2^m+1,5,0)}$ is $10$ when $m
      \equiv2~({\rm mod}~4)$ and $m\geq 6$.
  \item $\C_{(2,2^m+1,9,0)}=\C_{(2,2^m+1,10,0)}$ if $m\geq5$, and the minimum distance of $\C_{(2,2^m+1,9,0)}$ is $18$ when $m\equiv3~({\rm mod}~6)$ and $m\geq 9$.
\end{itemize}
\end{theorem}
\begin{proof}
Recall that $\beta$ is a primitive $(2^m+1)$-th root of unity in $\F_{2^{2m}}$, $f^*(x)$ denotes the reciprocal polynomial of $f(x)$ and $B_i^2$ denotes the number of codewords with Hamming weight $i$ of the binary Zetterberg code. Referring to \cite[Table 6.2]{Schoof}, we have that $B_2^2=B_4^2=0$ for any $l$ and $B_3^2=0$ if $l$ is even.

In the first case, from Theorem \ref{TH10}, we have $d\geq6.$ If $m$ is odd, then $3|2^m+1$. According to Theorem \ref{TH11}, $d=6$. If $m$ is even, then $3$ does not divide $2^m+1$. Let $f_1(x):=1+x+x^h$ with $3\leq h<2^{m}$ and $h\neq \frac{2^m+2}{2}$.
Note that $f_1(\beta)$ can not be zero because $B_3^2=0$ when $m$ is even. We can verify that $\frac{f_1^*(\beta)}{f_1(\beta)}$ is an element of $U_{2^m+1}$.
Then there exists $\beta^t\in U_{2^m+1}$ such that
\begin{equation}\label{beta}
  \frac{f_1(\beta)}{f_1^*(\beta)}=\frac{1+\beta+\beta^h}
  {1+\beta^{h-1}+\beta^h}=\beta^t.
\end{equation}
 We claim that $t\neq 0,1,h$. Assume that $t=0$. Then $\beta^{h-2}=1$, which is contrary to $3\leq h<2^{l}$. Assume that $t=1$. Then Eq.(\ref{beta}) is equivalent to  $\beta^{h+1}=1$, which is contrary to $h<2^{l}$.
 Assume that $t=h$. Then Eq.(\ref{beta}) is equivalent to $1+\beta+\beta^{2h-1}+\beta^{2h}=0$, which is contrary to $B_4^2=0.$ Notice that $\beta^{2h-1},\beta^{2h}\notin\{1,\beta\}$ because $h\neq \frac{2^l+2}{2}$. Therefore, $\#\{1,\beta,\beta^h,\beta^t,\beta^{h-1+t},\beta^{h+t}\}=6.$ Take $u_1=1$, $u_2=\beta$, $u_3=\beta^h$, $u_4=\beta^t$, $u_5=\beta^{h-1+t}$ and $u_6=\beta^{h+t}$,
 then $\sum_{i=0}^6u_i=0$, $d=6$.

In the second case, from Theorem \ref{TH10}, we have that $d\geq10.$ If $l=4t+2$, then $2^{4t+2}+1=(4+1)(4^{2t}+(-1)^1\cdot4^{2t-1}+\cdots+(-1)^{2t}\cdot1)$. According to Theorem \ref{TH11}, $d=10$.

In the third case, from Theorem \ref{TH10}, we have that $d\geq18.$ If $s=6t+3$, then $9|2^l+1$. According to Theorem \ref{TH11}, $d=18$.
\end{proof}
\begin{example}\label{example-03}
The following numerical examples were calculated by Magma programs, which coincide with Theorem \ref{COR11}.
\begin{itemize}
\item Let $(q,\delta)=(2,3)$ and $3\leq m\leq 14$, then the minimum distance of the code $\mathcal{C}_{(2,q^m+1,3,0)}$ is $6$.
\item Let $(q,\delta)=(2,5)$ and $m=6,10,14$, then the minimum distance of the code $\mathcal{C}_{(2,q^m+1,5,0)}$ is $10$.
\item Let $(q,\delta)=(2,9)$ and $m=9$, then the minimum distance of the code $\mathcal{C}_{(2,q^m+1,9,0)}$ is $18$.
\end{itemize}
\end{example}
\begin{remark}Since the length $n=q^m+1$ grows exponentially, there are few numerical examples we can verify. When $m=7,8,9,11,12,13$, the minimum distance of $\C_{(2,q^m+1,5,0)}$ is still $10$ but it is hard to prove that the minimum distance of $\C_{(2,q^m+1,5,0)}$ is $10$ for any $m$. Similarly, the minimum distance of $\mathcal{C}_{(2,q^m+1,9,0)}$ is $18$ if $m=10, 11$. There are too few numerical examples to conjecture that
the minimum distance of the code $\mathcal{C}_{(2,q^m+1,9,0)}$ is $18$ for any $m$.
\end{remark}

The next theorem partially proves Conjecture 23 in \cite{Li1}.
\begin{theorem}\label{COR144}
Let $q=3$. We have the following results.
\begin{itemize}
  \item The minimum distance of $\C_{(3,3^m+1,3,0)}$ is $4$ for any $m$.
  \item  $\C_{(3,3^m+1,4,0)}=\C_{(3,3^m+1,5,0)}$, and the minimum distance of $\C_{(3,3^m+1,4,0)}$ is $8$ if $m$ is odd or $m\equiv4~({\rm mod}~8)$ with $m\geq3$.
\end{itemize}
\end{theorem}
\begin{proof}
Let $\beta$ be a generator of $U_{3^l+1}$.
In the first case, from Lemma \ref{Li11}, we have  $d\geq4.$  
There exists a quadrinomial $g_1(x)=(x^{\frac{3^l+1}{2}}+1)\cdot(x^{\frac{3^l-1}{2}}-1)
=x^{3^l}-x^{\frac{3^l+1}{2}}
+x^{\frac{3^l-1}{2}}-1$ which satisfies $g_1(1)=1-1+1-1=0$ and $g_1(\beta)=\beta^{-1}-(-1)+(-\beta^{-1})-1=0$. Take $\{u_1,u_2,u_3,u_4\}=\{\beta^{-1},-1,-\beta^{-1},1\}$ and $\{c_1,c_2,c_3,c_4\}=\{1,-1,1,-1\}$. Then we have
\begin{equation*}
  \left\{
      \begin{array}{ll}
      c_1u_1^0+c_2u_2^0+c_3u_3^0+c_4u_4^0=0,\\
      c_1u_1+c_2u_2+c_3u_3+c_4u_4=0.
      \end{array}
    \right.
\end{equation*}
Therefore, the minimum distance of $\C_{(3,3^l+1,3,0)}$ is $4$.


In the second case, from Theorem \ref{TH10}, we have that $d\geq8.$
If $l$ is odd, then $4|3^l+1$. There exists an eight-term polynomial $g_2(x)=(x^{\frac{3^l+1}{2}}+1)\cdot(x^{\frac{3^l+1}{4}}+1)
\cdot(x^{\frac{3^l-3}{4}}-1)=x^{3^l}-x^{\frac{3(3^l+1)}{4}}
+x^{\frac{3\cdot3^l-1}{4}}-x^{\frac{3^l+1}{2}}+x^{\frac{3^l-1}{2}}
-x^{\frac{3^l+1}{4}}+x^{\frac{3^l-3}{4}}-1$ which satisfies $g_2(1)=0$ and $g_2(\beta)=g_2(\beta^2)=0$.


If  $l\equiv4~({\rm mod}~8)$, then $x^{82}-1|x^{3^l+1}-1$. There exists an eight-term polynomial $g_4(x)=x^{34}+x^{33}-x^{27}-x^{20}+x^{14}+x^7-x-1$ which satisfies $g(x)|x^{82}-1$, $g_4(1)=0$, $g_4({\beta^{\frac{q+1}{82}}})=0$ and $g_4({\beta^{\frac{q+1}{41}}})=0$. We can verify that $\beta^{\frac{q+1}{82}}$ and $\beta^{\frac{q+1}{41}}$ are roots of $g_4(x)$ by Magma.
\end{proof}
\begin{remark}
Computing by Magma programs, the minimum distance of $\C_{(3,3^m+1,5,0)}$ is $8$ when $m=6$. There are too few numerical examples to conjecture that
the minimum distance of the code $\mathcal{C}_{(3,q^m+1,4,0)}$ is $8$ for any $m$.
\end{remark}
With similar discussion as in Theorem \ref{COR144}, it is easy to get the following result.
\begin{corollary}\label{xiu3}
The minimum distance of $\C_{(q,q^m+1,3,0)}$ is $4$ if $q>2$ and $m\geq3$.
\end{corollary}
\begin{example}\label{example-06}
We have the following examples from Corollary \ref{xiu3}.
\begin{itemize}
\item Let $q=3$ and $m=3$, then the code $\mathcal{C}_{(3,28,3,0)}$ has parameters $[28,21,4]$.
\item Let $q=5$ and $m=3$, then the code $\mathcal{C}_{(5,126,3,0)}$ has parameters $[126,119,4]$.
\end{itemize}
These  codes are optimal according to the tables of best known codes in \cite{Gra}. These results are verified by Magma programs.
\end{example}


Let us focus on the case where $b=1$ in the last part.
From \cite{Dobbertin06}, it is easy to get the following lemma.
\begin{lemma}\label{lambda}
Let $q=2^t$, $t\geq 1$ and $\mathcal{H}=\{x \in \F_{q^m}^*\,:\,{\rm Tr}_2^{q^m}(x^{-1})=1\}$,  then $\mathcal{H}$ can be written as
$$\mathcal{H}=\{\lambda+{\lambda}^{-1}\,:\,\lambda \in U_{q^m+1}\setminus \{1\}\}.$$
\end{lemma}
\begin{theorem}\label{Theorem-05}
Let $m>3$, then the code $\C_{(q,q^m+1,2,1)}$ has parameters $[q^m+1,q^m+1-2m,d]$, where
\begin{eqnarray*}
d=\begin{cases}
2,& \text{ if $q$ is odd},\\
3,& \text{ if $Char(\F_q)=2$ and $m$ is odd},  \\
4,& \text{ if $Char(\F_q)=2$, $q>2$ and $m$ is even}, \\
5,& \text{ if $q=2$ and $m\equiv 2 \pmod 4$},  \\
5\,\,\text{or}\,\, 6, & \text{ if $q=2$ and $m\equiv0 \pmod 4$}.
\end{cases}
\end{eqnarray*}
\end{theorem}
\begin{proof}
The dimension of the code follows from the fact $|C_1|=2m$. It is easily seen that $d>1$. Assume that there is a
 codeword $\mathbf{c}=(c_0,c_1,\ldots,c_{q^m})\in \C_{(q,q^m+1,2,1)}$, then
$$\sum_{i=0}^{n-1}c_i\beta^i=0.$$
When $q$ is odd, it is obvious that $d=2$ since $1+\beta^{\frac{q^m+1}{2}}=0$. We next prove the results for $q$ being
 an even prime power from the following two cases.

\noindent{\bf Case 1:} $q=2$. From the BCH bound and the sphere-packing bound, we obtain that
$3\leq d \leq 6.$
If $d=3$, then there exists $\{u_1,u_2\}\in{U_{2^m+1}\choose 2}\setminus\{1\}$ such that
\begin{equation}\label{xiu1}
1+u_1+u_2=0.
\end{equation}
Raising both sides of Eq.(\ref{xiu1})
to $(2^m+1)$-th power, we have
\begin{equation*}
1+u_1+u_1^{-1}=0.
\end{equation*} From Lemma \ref{lambda}, there exists $u_1\in U_{2^m+1}\setminus\{1\}$ such that $u_1+u_1^{-1}=1$ if and only if ${\rm Tr}_2^{2^m}(1)=1$. Hence,
 $d=3$ if and only if $m$ is odd. Otherwise, we have $d>3$.

If there is a codeword in $\C_{(2,2^m+1,2,1)}$ with Hamming weight $4$, then there exists $\{u_1,u_2,u_3\}\\ \in{U_{2^m+1}\choose 3}\setminus\{1\}$ such that
\begin{equation}\label{beta2}
1+u_1+u_2+u_3=0.
\end{equation}
Raising both sides of Eq.(\ref{beta2})
to $(2^m+1)$-th power, then
\begin{equation*}
\begin{split}
1&=\left(1+u_1^{-1}+u_2^{-1} \right)\left(1+u_1+u_2\right)\\
&=1+{\rm Tr}_{2^{\frac{m}{2}}}^{2^m}\left(u_1+u_2\right)+u_1u_2^{-1}+u_1^{-1}u_2,
\end{split}
\end{equation*}
which means that
\begin{equation}\label{beta3}
{\rm Tr}_{2^{\frac{m}{2}}}^{2^m}\left(u_1+u_2\right)={\rm Tr}_{2^{\frac{m}{2}}}^{2^m}\left(u_1u_2^{-1}\right).
\end{equation}
 From Eq.(\ref{beta2}) we have
\begin{equation}\label{beta4}
{\rm Tr}_{2^{\frac{m}{2}}}^{2^m}\left(u_1+u_2\right)={\rm Tr}_{2^{\frac{m}{2}}}^{2^m}\left(1+u_3\right)={\rm Tr}_{2^{\frac{m}{2}}}^{2^m}\left(u_3\right).
\end{equation}
Combining Eqs.(\ref{beta3}) and (\ref{beta4}), it is easily seen that
$$u_1u_2^{-1}+u_1^{-1}u_2=u_3+u_3^{-1},$$
which implies that $u_3=u_1u_2^{-1}$ or $u_3=u_1^{-1}u_2$. Without loss of generality, we assume that $u_3=u_1u_2^{-1}$. Then Eq.(\ref{beta2}) becomes
$\left(u_1+u_2\right)\left(1+u_2^{-1}\right)=0,$
contradicting our assumption. Hence, $d=5$ or $6$ if $m$ is even.

If $m\equiv 2 \pmod 4$, it is clear that $5 \, | \, 2^m+1$. Then
$$0=1-\beta^{2^m+1}=\left(1-\beta^{\frac{2^m+1}{5}}\right)\left(1+ \beta^{\frac{2^m+1}{5}}+ \beta^{\frac{2^{m+1}+2}{5}}+ \beta^{\frac{3\cdot2^m+3}{5}}+ \beta^{\frac{2^{m+2}+4}{5}}\right),$$
which means that $1+ \beta^{\frac{2^m+1}{5}}+ \beta^{\frac{2^{m+1}+2}{5}}+ \beta^{\frac{3\cdot2^m+3}{5}}+
\beta^{\frac{2^{m+2}+4}{5}}=0$ since $1-\beta^{\frac{2^m+1}{5}}\neq 0$. Hence, $d= 5$ if $m\equiv 2 \pmod 4$.

\noindent{\bf Case 2:} $q\neq 2$ is an even prime power. From the BCH bound and the sphere-packing bound, we can check that
\begin{equation}\label{eq:d34}
3\leq d \leq 4.
\end{equation}
If $d=3$, then there exist $a_0$, $a_1\in \F_q^*$ and $\{u_1,u_2\}\in{U_{q^m+1}\choose 2}\setminus\{1\}$ such that
\begin{equation}\label{eq:beta-0}
a_0+a_1{u_1}+{u_2}=0.
\end{equation}
If $m$ is odd, let $a_0=a_1=1$, then Eq.(\ref{eq:beta-0}) becomes $1+u_1={u_2}.$ Repeating the steps above, we have
\begin{equation*}
1+{u_1}+{u_1}^{-1}=0.
\end{equation*} From Lemma \ref{lambda}, there exists $u_1\in U_{q^m+1}\setminus\{1\}$
such that $u_1+{u_1}^{-1}=1$. Hence, $d=3$.

If $m$ is even, from Eq.(\ref{eq:beta-0}) we have
$$\left(a_0+a_1{u_1}\right)^{q^m+1}=\left({u_2}\right)^{q^m+1}=1,$$
which implies that
\begin{equation}\label{eq:beta}
{u_1}+{u_1}^{-1}=\frac{1+a_0^2+a_1^2}{a_0a_1}.
\end{equation}
It is clear that for any $a_0,a_1\in \F_q^*$, let $q=2^t$, we have
$${\rm Tr}_2^{2^{tm}}\left(\frac{1+a_0^2+a_1^2}{a_0a_1}\right)=m\cdot{\rm Tr}_2^{2^t}\left(\frac{1+a_0^2+a_1^2}{a_0a_1}\right)=0.$$
From Lemma \ref{lambda}, there does not exist $u_1\in U_{q^m+1}\setminus\{1\}$ such that Eq.(\ref{eq:beta}) holds.
Combining this with Eq.(\ref{eq:d34}), we have $d=4$ if $m$ is even.
\end{proof}

\begin{example}\label{example-04}
The following numerical examples were calculated by Magma programs, which coincide with Theorem \ref{Theorem-05}.
\begin{itemize}
\item Let $q=2$ and $m=4$, then the code $\mathcal{C}_{(2,17,2,1)}$ has parameters $[17,9,5]$.
\item Let $q=2$ and $m=6$, then the code $\mathcal{C}_{(2,65,2,1)}$ has parameters $[65,53,5]$.
\item Let $q=4$ and $m=2$, then the code $\mathcal{C}_{(4,17,2,1)}$ has parameters  $[17,13,4]$.
\end{itemize}
\end{example}
\begin{remark}
In \cite[Theorem 18]{Gra}, the authors showed that the minimum distance of $\C_{(q,q^m+1,2,1)}$ is at least $2$.
 In Theorem \ref{Theorem-05}, we give the more precise Hamming distance of $\C_{(q,q^m+1,2,1)}$.
\end{remark}

The next corollary partially proves the open problem about the minimum distance of $\C_{(q,q+1,3,1)}|_{\F_3}$ which is mentioned in the last paragraph in \cite[Section 4]{D5}.
\begin{theorem}\label{COR15}
Let $m$ be a positive integer. Then the code $\C_{(3,3^m+1,3,1)}$ has parameters $[3^m+1,3^m+1-4m,d]$, where
\begin{equation*}
  d=\left\{
      \begin{array}{ll}
        4, & \hbox{$m$~~is~~odd,} \\
        5, & \hbox{$m\equiv2~({\rm mod}~4)$,} \\
        6, & \hbox{$m\equiv4~({\rm mod}~8)$,}\\
        \geq6, & \hbox{$m\equiv0~({\rm mod}~8)$}.
      \end{array}
    \right.
\end{equation*}
\end{theorem}
\begin{proof}
The dimension of $\C_{(3,3^m+1,3,1)}$ has been given in \cite[Theorem 28]{D5}. According to Theorem \ref{TH10}, the minimum distance $d$ of
$\C_{(3,3^m+1,3,1)}$ is no less than $4$. For convenience, we use the notation $\sigma_{i,j}$ to denote the $\sigma_{i,j}(u_1,\ldots,u_i)$ in this proof.

If $m$ is odd, then $3^m+1=3^{2t+1}+1=(3+1)(3^{2t}-3^{2t-1}+\cdots-3+1)$. From Corollary \ref{COR12}, we have $d=4.$

If $m$ is even, then there is no element with order $4$ in $U_{3^m+1}$. For the sake of convenience, put $q=3^m$ in this proof. Firstly, we claim that $d\geq5$ in this case. Assume that $d=4$. Then there exists $\ccc^{\prime}=(c_1,\ldots,c_{q+1})\in\F_3^{q+1}$ such that $\ccc^{\prime}H^T=\mathbf{0}$, where $wt(\ccc^{\prime})=4$. Let $supp(\ccc^{\prime})=\{i_1,i_2,i_3,i_4\}$, then $c_{i_1},c_{i_2},c_{i_3},
c_{i_4}\in \{1,-1\}\subseteq U_{q+1}$. It then follows that
\begin{equation*}\label{ddd}
  \left\{
    \begin{array}{ll}
      c_{i_1}u_1+c_{i_2}u_2+c_{i_3}u_3+c_{i_4}u_4=0,\\
      c_{i_1}u_1^2+c_{i_2}u_2^2+c_{i_3}u_3^2+c_{i_4}u_4^2=0.
    \end{array}
  \right.
\end{equation*}
By the symmetry of $c_{i_1},c_{i_2},c_{i_3}$ and $c_{i_4}$, we only need to consider the following cases:
\begin{eqnarray*}
  {\rm(i)}~\{c_{i_1},c_{i_2},c_{i_3},c_{i_4}\}&=&\{1,1,1,1\}; \\
 {\rm(ii)}~ \{c_{i_1},c_{i_2},c_{i_3},c_{i_4}\}&=&\{1,1,1,-1\}; \\
 {\rm(iii)}~\{c_{i_1},c_{i_2},c_{i_3},c_{i_4}\}&=&\{1,1,-1,-1\}.
\end{eqnarray*}
In case {(i)}, we have that
\begin{equation*}\label{ddd}
  \left\{
    \begin{array}{ll}
      u_1+u_2+u_3+u_4=0,\\
      u_1^2+u_2^2+u_3^2+u_4^2=0.
    \end{array}
  \right.
\end{equation*}
Then $\sigma_{4,1}=\sigma_{4,2}=\sigma_{4,3}=0$. Let $\sigma_{x}=(x-u_1)(x-u_2)(x-u_3)(x-u_4)$. It is easy to obtain that $\sigma{(x)}=x^4-\sigma_{4,4}$, and $\sigma{(x)}$ can not have four zeroes in $U_{q+1}$ as $\gcd(3^m+1,4)=2$.

In case (ii), we obtain that $\sum_{i=1}^3u_i=u_4$ and $\sum_{i=1}^3u_i^2=u_4^2$. Combining the previous two equations, we have that $u_3\sigma_{2,1}+\sigma_{2,2}=0$. It then follows from $u_3\in U_{q+1}$ that $u_1^2+u_2^2+u_1u_2=(u_1-u_2)^2=0$, a contradiction!

In case (iii), we need to check that there is no $\{u_1,u_2,u_3,u_4\}\in {U_{q+1}\choose 4}$ such that $u_1+u_2-(u_3+u_4)=0$ and $u_1^2+u_2^2-(u_3^2+u_4^2)=0$.
Combining the previous two equations, we have $u_3^2-(u_1+u_2)u_3+u_1u_2=(u_3-u_2)(u_3-u_1)=0$, which is contrary to our assumption that $u_1,u_2$ and $u_3$ are pairwise distinct.

In summary, $d>4$ when $m$ is even. In the even case, we have two more subcases.

\begin{itemize}
  \item If $m=4t+2$, then $3^{4t+2}+1=(9+1)(9^{2t}-9^{2t-1}+\cdots-9+1)$. It then follows from $5|q+1$ that there exists $\xi_1$ in $U_{q+1}$ which has order $5$. Let $u_1=\xi_1$, $u_2=\xi_1^2$, $u_3=\xi_1^3$, $u_4=\xi_1^4$ and $u_5=\xi_1^5$. It is not difficult to check that $\sum_{i=1}^5u_i=0$ and $\sum_{i=1}^5u_i^2=0.$
  \item If $m=4t$, then there is no element with order $5$ in $U_{q+1}$. We claim that $d> 5$ in this subcase. Assume that $d=5$.
Then there exist $\{u_1,u_2,u_3,u_4,u_5\}\in U_{q+1}$ and $(c_{i_1},c_{i_2},c_{i_3},c_{i_4},c_{i_5})\in(\F_3^*)^5$ such that
\begin{equation*}\label{fff}
  \left\{
    \begin{array}{ll}
      c_{i_1}u_1+c_{i_2}u_2+c_{i_3}u_3+c_{i_4}u_4+c_{i_5}u_5=0,\\
      c_{i_1}u_1^2+c_{i_2}u_2^2+c_{i_3}u_3^2+
c_{i_4}u_4^2+c_{i_5}u_5^2=0.
    \end{array}
  \right.
\end{equation*}
By the symmetry of $c_{i_1},c_{i_2},c_{i_3},c_{i_4}$ and $c_{i_5}$, we only need to consider the following three cases:
\begin{eqnarray*}
  {\rm(i)}~\{c_{i_1},c_{i_2},c_{i_3},c_{i_4},c_{i_5}\}&=&\{1,1,1,1,1\}; \\
 {\rm(ii)}~ \{c_{i_1},c_{i_2},c_{i_3},c_{i_4},c_{i_5}\}&=&\{1,1,1,1,-1\}; \\
 {\rm(iii)}~\{c_{i_1},c_{i_2},c_{i_3},c_{i_4},c_{i_5}\}&=&\{1,1,1,-1,-1\}.
\end{eqnarray*}
In case (i), we have that

\begin{equation*}\label{ggg}
\left\{
  \begin{array}{lll}
    u_3+u_4+u_5=-\sigma_{2,1}, \\
    u_3u_4+u_3u_5+u_4u_5={\sigma_{2,1}}^2-\sigma_{2,2},\\
    u_3u_4u_5=-(\sigma_{2,1}^3+\sigma_{2,1}\sigma_{2,2}).
  \end{array}
\right.
\end{equation*}

It then follows from $u_3u_4u_5\in U_{q+1}$ that $(u_3u_4u_5)^q=(u_3u_4u_5)^{-1}$. Then $(\sigma_{2,1}^3+\sigma_{2,1}\sigma_{2,2})^2=\sigma_{2,2}^3$, which is the same as
\begin{equation}\label{zzz}
  (u_1-u_2)(u_1^5-u_2^5)=0.
\end{equation} Eq.(\ref{zzz}) holds only if $u_1^5=u_2^5$.
Since $m=4t$ and $\gcd(5,q+1)=1$, it then follows from  $u_1^5=u_2^5$ that $u_1=u_2$, a contradiction!

In case (ii), we observe that
\begin{equation*}\label{ggg}
\left\{
  \begin{array}{lll}
    u_1+u_2+u_3+u_4=u_5=\sigma_{4,1}(u_1,u_2,u_3,u_4), \\
    u_1^2+u_2^2+u_3^2+u_4^2=u_5^2=\sigma_{4,1}^2(u_1,u_2,u_3,u_4)-
    2\sigma_{4,2}(u_1,u_2,u_3,u_4), \\
    u_1^{-1}+u_2^{-1}+u_3^{-1}+u_4^{-1}=u_5^{-1}=\sigma_{4,1}^q
    (u_1,u_2,u_3,u_4)=\frac{\sigma_{4,3}(u_1,u_2,u_3,u_4)}
    {\sigma_{4,4}(u_1,u_2,u_3,u_4)}.
  \end{array}
\right.
\end{equation*}
Then $\sigma_{4,2}(u_1,u_2,u_3,u_4)=0$ and $\sigma_{4,2}(u_1,u_2,u_3,u_4)=u_5\sigma_{4,3}(u_1,u_2,u_3,u_4).$
Let $\sigma{(x)}=x^4-u_5x^3-\sigma_{4,3}(u_1,u_2,u_3,u_4)x+u_5
\sigma_{4,3}(u_1,u_2,u_3,u_4)$. From the above fact, we can obtain $\sigma(x)=(x^3-\sigma_{4,3}(u_1,u_2,u_3,u_4))(x-u_5)$, which can not have $4$ zeroes in $U_{q+1}$.

In case (iii), we have that
\begin{equation*}\label{ggg}
\left\{
  \begin{array}{lll}
    u_1+u_2+u_3+u_4+u_5=0=\sigma_{5,1}, \\
    u_1^2+u_2^2+u_3^2+u_4^2+u_5^2=0=\sigma_{5,1}^2-
    2\cdot\sigma_{5,2}, \\
    u_1^{-1}+u_2^{-1}+u_3^{-1}+u_4^{-1}+u_5^{-1}=0=
    \frac{\sigma_{4,1}}{\sigma_{5,1}}.
  \end{array}
\right.
\end{equation*}
Hence, $\sigma_{5,1}=\sigma_{5,2}=\sigma_{5,3}=\sigma_{5,4}=0$.
Let $\sigma(x)=(x-u_1)(x-u_2)(x-u_3)(x-u_4)(x-u_5)$. Similarly, we have $\sigma(x)=x^5+\sigma_{5,1},$ which can not have $4$ zeroes in $U_{q+1}$ since $\gcd(q+1,5)=1.$

Therefore, $d>5$ when $m=4t$.
\end{itemize}

If $m\equiv4~({\rm mod}~8)$, then $x^{82}-1|x^{3^m+1}-1$. There exists a six-term factor $f(x)=x^{16}-x^{15}+x^{11}+x^5-x+1$ which satisfies $f(x)|x^{82}-1$ and $f(\beta^{\frac{5(q+1)}{82}})=f(\beta^{\frac{5(q+1)}{41}})=0$, where $\beta$ is a generator of $U_{q+1}$. The fact that $\beta^{\frac{5(q+1)}{82}}$ and $\beta^{\frac{5(q+1)}{41}}$ are roots of $f(x)$ was checked by Magma.

This completes the proof.
\end{proof}
\begin{example}\label{example-04}
We have the following examples for the code of Theorem \ref{COR15}.
\begin{itemize}
\item Let $m=2$, then the code $\mathcal{C}_{(3,10,3,1)}$ has parameters $[10,2,5]$.
\item Let $m=3$, then the code $\mathcal{C}_{(3,28,3,1)}$ has parameters $[28,16,4]$.
\item Let $m=4$, then the code $\mathcal{C}_{(3,82,3,1)}$ has parameters  $[82,66,6]$.
\item Let $m=5$, then the code $\mathcal{C}_{(3,244,3,1)}$ has parameters  $[244,244,4]$.
\item Let $m=6$, then the code $\mathcal{C}_{(3,730,3,1)}$ has parameters  $[730,706,5]$.
\end{itemize}
The code $\mathcal{C}_{(3,244,3,1)}$ are distance-optimal cyclic codes when $m=2$ and $m=3$.
These results are verified by Magma programs.
\end{example}




\begin{remark}
 (i) From the above examples in this paper, the codes $\C_{(q,q^l+1,\delta,0)}$ and $\C_{(q,q^l+1,\delta,1)}$ are sometimes dimension-optimal or best possible for cyclic codes according to \cite[Appendix A]{D1} and \cite{Gra} when $m$ and $\delta$ are small. Some examples which have nice parameters were given in \cite[Example 20, Example 22]{Li1}.

 (ii) Conjecture 23 in \cite{Li1} said that the minimum distance of $\C_{(3,3^l+1,5,0)}$ is $8$, which is supported by Magma experiments. In this paper, we have proved that $d=8$ when $l$ is odd or $l\equiv4~({\rm mod}~8)$ in Theorem \ref{COR144}, the remaining case is still open.
When $l$ is odd, $l\equiv2~({\rm mod}~4)$ or $l\equiv4~({\rm mod}~8)$, the minimum distance of $\C_{(3,3^l+1,3,1)}$ are given in Theorem  \ref{COR15}.  We do not have enough numerical examples to conjecture that $d=6$ as Magma program is hard to compute the minimum distance of $\C_{(3,3^l+1,3,1)}$ when $l\equiv0~({\rm mod}~8)$. The reader is cordially invited to attack these two remaining cases.
\end{remark}

\section{Conclusion}
The main contributions of this paper are the following:
\begin{enumerate}
  \item [(1)] We provided a necessary and sufficient condition for $0\leq a\leq q^m$ being a coset leader modulo $n$ for any prime power $q$.
  \item [(2)] A lower bound on the minimum distance of $\C_{(q,q+1,\delta,1)}$ is given. 
 Some improved lower bounds on the minimum distance of $\C_{(q,q^m+1,\delta,0)}$ and $\C_{(q,q^m+1,\delta,b)}$ are given. 
  \item [(3)] Two open problems in \cite{D5,Li1} are partially solved in this paper. The method of using the ESPs $\sigma_{i,j}(u_1,\ldots,u_j)$ to study the minimum distance of BCH codes is novel and innovative.
  \item [(4)] According to the tables of best codes known in~\cite{Gra}, some of the proposed linear codes are optimal.
\end{enumerate}


\begin{thebibliography}{99}

\bibitem{Aug}D. Augot, P. Charpin, N. Sendrier, Studying the locator polynomials of minimum weight codewords of BCH codes, IEEE Trans. Inf. Theory, {\bf 38} (3), (1992), pp. 960--973.
    
    
\bibitem{Augot}D. Augot, N. Sendrier, Idempotents and the BCH bound, IEEE
Trans. Inf. Theory,  {\bf 40} (1), (1994), pp. 204--207.


\bibitem{Betten}A. Betten, M. Braun, H. Fripertinger, A. Kerber, A. Kohnert,
A. Wassermann, Error-Correcting Linear Codes. Berlin, Germany:
Springer-Verlag, (2006).


\bibitem{Ber}E. R. Berlekamp, The enumeration of information symbols in BCH codes, Bell Syst. Tech. J.,  {\bf46} (8), (1967), pp. 1861--1880.
    

\bibitem{Bose}R. C. Bose, D. K. Ray-Chaudhuri, On a class of error correcting binary group codes, Inf. Control, {\bf 3} (1), (1960), pp. 68--79.
    
    
\bibitem{Car}C. Carlet, S. Guilley, Complementary dual codes for counter-measures to side-channel attacks, in:
E.R. Pinto, et al. (Eds.), Coding Theory and Applications, in: CIM Series in Mathematical Sciences,
Springer Verlag, {\bf 3}, (2014), pp. 97--105, Adv. Math. Commun., {\bf 10} (1), (2016), pp. 131--150.


\bibitem{Charp} P. Charpin, Open problems on cyclic codes, in V.S. Pless, W.C. Huffman (Eds.), Handbook of coding Theory, vol. I, North-Holland, (1998), pp. 963--1063 (Chapter 11).
    
    
\bibitem{Charp1}  P. Charpin, T. Helleseth, V. A. Zinoviev, The coset distribution of triple-error-correcting binary primitive BCH codes, IEEE Trans. Inf. Theory, {\bf 52} (4), (2006), pp. 1727--1732.


\bibitem{DELSARTE} P. Delsarte, On subfield subcodes of modified Reed-Solomon codes, IEEE Trans. Inform. Theory, {\bf 21} (5), (1975), pp. 575--576.
    
    
\bibitem{DingA}  C. Ding, Parameters of several classes of BCH codes, IEEE Trans. Inf. Theory, {\bf 61} (10), (2015), pp. 5322--5330.
    


\bibitem{D1} C. Ding, Codes from Difference Sets, World Scientific, Singapore, (2015).
    
    
\bibitem{D2} C. Ding, Designs from Linear Codes, World Scientific, Singapore, (2018).
    
    
\bibitem{D55} C. Ding, BCH codes in the past 55 years, The 7th international Workshop on Finite Fields Applications, Tianjin, China, (2016).
    
    
\bibitem{DingL10} C. Ding, X. Du, Z. Zhou, The Bose and minimum distance
of a class of BCH Codes, IEEE Trans. Inf. Theory, {\bf61} (5), (2015), pp. 2351--2356.


\bibitem{DingB}C. Ding, C. L. Fan, Z. C. Zhou, The dimension and minimum distance of two classes of primitive BCH codes, Finite Fields Appl., {\bf 45}, (2017), pp. 237--263.



\bibitem{D5} C. Ding, C. Tang, Infinite families of near MDS codes holding $t$-designs, IEEE Trans. Inform. Theory, {\bf 66} (9), (2020), pp. 5419--5428.

\bibitem{Desa} Y. Desaki, T. Fujiwara, T. Kasami, The weight distributions of extended binary primitive BCH codes of length $128$, IEEE Trans. Inf. Theory, {\bf 43} (4), (1997), pp. 1364--1371.
    
    
\bibitem{Dobbertin06} H. Dobbertin, P. Felke, T. Helleseth, P. Rosendahl, Niho type cross-correlation functions via Dickson
 polynomials and Kloosterman sums, IEEE Trans. Inf. Theory, {\bf52} (2), (2006), pp. 613--627.


\bibitem{Fujiwa}T. Fujiwara, T. Takata, T. Kasami, S. Lin, An approximation to the weight distribution of binary primitive BCH codes with designed distances 9 and 11, IEEE Trans. Inf. Theory, {\bf 32} (5), (1986),
pp. 706--709.

\bibitem{Geer} G. van der Geer, M. van der Vlugt, Artin-Schreier curves and
codes, J. Algebra, {\bf 139} (1), (1991), pp. 256--272.


\bibitem{Gra} M. Grassl, Code Tables, http://www.codetables.de.


\bibitem{Hocq} A. Hocquenghem, Codes correcteurs d'erreurs, Chiffres, {\bf 2} (2), (1959), pp. 147--156.

\bibitem{Hel} H. Helgert, R. Stinaff, Shortened BCH codes, IEEE Trans. Inf. Theory, {\bf 19} (6), (1973), pp. 818--820.


\bibitem{Kra}I. Krasikov, S. Litsyn, On the distance distributions of BCH codes and their duals, Des., Codes Cryptogr., {\bf 23} (2), (2001), pp. 223--232.
    
    
\bibitem{Ker}O. Keren, S. Litsyn, More on the distance distribution of BCH codes, IEEE Trans. Inf. Theory, {\bf 45} (1), (1999), pp. 251--255.
    
    
\bibitem{Kas}T. Kasami, Weight distributions of Bose-Chaudhuri-Hocquenghem codes, in Combinatorial Mathematics and Applications, R. C. Bose
and T. A. Dowlings, Eds. Chapel Hill, NC, USA: Univ. North Carolina Press, 1969, ch. 20.


\bibitem{Kas1}T. Kasami, T. Fujiwara, S. Lin, An approximation of the weight
distribution of binary linear codes, IEEE Trans. Inf. Theory,  {\bf 31} (6), (1985),  pp. 769--780.


\bibitem{Kas2}T. Kasami, S. Lin, Some results on the minimum weight of primitive BCH codes, IEEE Trans. Inf. Theory, {\bf 18} (6), (1972), pp. 824--825.
    
    
\bibitem{Kas3}T. Kasami, S. Lin, W. W. Peterson, Linear codes which are invariant under the affine group and some results on minimum weights in BCH
codes, Electron. Commun. Jpn., {\bf 50} (9), (1967), pp. 100--106.


\bibitem{Li1} C. Li, C. Ding, LCD cyclic codes over finite fields, IEEE Trans. Inform. Theory, {\bf 63} (7), (2017), pp. 4344--4356.
    
    
\bibitem{LiL18} S. Li, C. Li, C. Ding, H. Liu, Two families of LCD BCH codes,
IEEE Trans. Inf. Theory, {\bf63} (9), (2017), pp. 5699--5717.



\bibitem{Liu1} H. Liu, C. Ding, S. Li, Dimensions of three types of BCH codes over $\F_q$, Discrete Math., {\bf 340}, (2017), pp. 1910--1927.
    
    
\bibitem{LY1} Y. Liu, R. Li, Q. Fu, L. Lu, Y. Rao, Some binary BCH codes with length $n=2^m+1$, Finite Fields Appl., {\bf 55}, (2019), pp. 109--133.
    
    
\bibitem{LY2} Y. Liu, R. Li, L. Guo, H. Song, Dimensions of nonbinary antiprimitive BCH codes and
some conjectures, arXiv: 1712.06842v2.


\bibitem{MacW} F. J. MacWilliams, N. J. A. Sloane, The Theory Error-Correcting Codes (North-Holland Mathematical Library). Amsterdam, The Netherlands: North-Holland, (1997).
    
    
\bibitem{Muir} T. Muir, A Treatise on the Theory of Determinants, Dover, New York, (1960).
    
    
\bibitem{Schoof} R. Schoof, M. V. D. Vlugt, Hecke operators and the weight distributions of certain codes, J. Combin. Theory A, {\bf 57} (2), (1991), pp. 163--186.
    
    
\bibitem{Tang1} C. Tang, C. Ding, An infinite family of linear codes supporting $4$-designs, IEEE Trans. Inform. Theory,  {\bf 67} (1), (2021), pp. 244--254.
    
    


\bibitem{HVDV} H. Van de Vel, Numerical treatment of a generalized Vandermonde system of equations, Linear Algebra Appl., {\bf 17}, (1977), pp. 149--179.
    
    
\bibitem{Xia1} Y. Xia, N. Li, X. Zeng, T. Helleseth, On the correlation distribution for
a Niho decimation, IEEE Trans. Inform. Theory, {\bf 63} (11), (2017), pp. 7206--7218.


\bibitem{YanL24} H. Yan, H. Liu, C. Li, S. Yang, Parameters of LCD BCH codes with two lengths, Adv. Math. Commun., {\bf12} (3),
(2018), pp. 579--594.


\bibitem{Yang} X. Yang, J. L. Massey, The condition for a cyclic code to have a complementary dual. Discrete Math., {\bf 126} (1--3), (1994), pp. 391--393.
    
    
\bibitem{YueL26} D. Yue, Z. Feng, Minimum cyclotomic coset representatives
and their applications to BCH codes and Goppa codes, IEEE Trans.
Inf. Theory, {\bf46} (7), (2000), pp. 2625--2628.


\bibitem{YueL27} D. Yue, Z. Hu, On the dimension and minimum
distance of BCH codes over $\F_q$, J. Electron., {\bf13} (3), (1996), pp. 216--221.


\bibitem{Zhu}S. Zhu, Z. Sun, X. Kai, A class of narrow-sense BCH
codes, IEEE Trans. Inform. Theory, {\bf 65} (8), (2019), pp. 4699--4714.


\end{thebibliography}
\end{document}